\documentclass[acmsmall, nonacm]{acmart}

\acmConference
\acmBooktitle

\setcopyright{none}

\usepackage{thm-restate}
\usepackage{amsmath, amsthm}
\usepackage{mathtools}
\usepackage{cleveref}
\usepackage{xspace}
\usepackage[dvipsnames,table]{xcolor} 
\usepackage{algpseudocode}
\usepackage{tcolorbox}
\usepackage{tikz}
\usetikzlibrary{shapes, positioning}

\tcbset{
  base/.style={
    arc=0mm, 
    bottomtitle=0.5mm,
    boxrule=0mm,
    colbacktitle=black!10!white, 
    coltitle=black, 
    fonttitle=\bfseries, 
    left=2.5mm,
    leftrule=1mm,
    right=3.5mm,
    title={#1},
    toptitle=0.75mm, 
  }
}

\newtcolorbox{mainbox}[1]{
  colframe=black!40, 
  base={#1}
}

\newtcolorbox{subbox}[1]{
  colframe=black!30!white,
  base={#1}
}

\newcommand{\defeq}{\vcentcolon=}

\DeclareMathOperator{\avg}{avg}

\DeclarePairedDelimiter{\multiset}{\lbrace \!\! \lbrace}{\rbrace \!\!\rbrace}
\DeclarePairedDelimiter{\size}{\lvert}{\rvert}
\DeclarePairedDelimiter{\set}{\lbrace}{\rbrace}
\DeclareMathOperator{\Shapley}{\ensuremath{\mathrm{Shapley}}}
\newcommand{\mult}[2]{\ensuremath{m[#1\!#2]}}

\def\finbags{\mathcal{B}_{\mathsf{fin}}}
\def\e#1{\emph{#1}}
\def\set#1{\mathord{\{#1\}}}
\def\naturals{\mathbb{N}}
\def\reals{\mathbb{R}}

\def\ra{\rightarrow}
\def\la{\leftarrow}

\def\fpsharpp{${\mathrm{FP}^{\mathrm{\#P}}}$\xspace}

\def\consts{\mathsf{Const}}

\def\scs{\mathcal{S}}

\def\arity{\mathit{ar}}

\def\aggf#1{\mathord{\mathsf{#1}}}
\def\aggsum{\aggf{Sum}}
\def\aggcount{\aggf{Count}}
\def\aggavg{\aggf{Avg}}
\def\aggmin{\aggf{Min}}
\def\aggmax{\aggf{Max}}
\def\aggmed{\aggf{Med}}
\def\aggcdis{\aggf{CDist}}
\def\aggquantile{\aggf{Qnt}}
\def\aggduplicates{\aggf{Dup}}
\def\aggnoduplicates{\aggf{NoDup}}
\def\hom{\mathit{Hom}}

\def\nusetcover{\mathord\nu_{\mathsf{sc}}}

\def\Qbool{Q_{\mathsf{bool}}}

\def\vars{\mathord{\mathit{vars}}}
\def\freevars{\mathord{\mathit{vars}_{\mathsf{F}}}}
\def\exvars{\mathord{\mathit{vars}_{\exists}}}
\def\atoms{\mathrm{atoms}}

\def\cupdot{\mathbin{\dot{\cup}}}

\def\exhierarchical{$\exists$-hierarchical\xspace}
\def\allhierarchical{all-hierarchical\xspace}
\def\qhierarchical{q-hierarchical\xspace}
\def\sqhierarchical{sq-hierarchical\xspace}

\def\exo{^{\mathsf{x}}}
\def\endo{^{\mathsf{n}}}

\def\X{\mathcal{X}}
\def\C{\mathcal{C}}
\def\Y{\mathcal{Y}}
\def\P{\mathcal{P}}

\def\Qxyy{Q_{xyy}}
\def\Qxyyz{Q_{xyyz}}
\def\QFxyy{Q_{xyy}^{\mathrm{full}}}
\def\tauid{\tau_{\mathsf{id}}}
\def\tauidi#1{\tau^{#1}_{\mathsf{id}}}
\def\taugb#1{\tau_{>{#1}}}
\def\taugbi#1#2{\tau^{#2}_{>#1}}
\def\taurelui#1{\tau^{#1}_{\mathsf{ReLU}}}
\def\taurelu{\tau_{\mathsf{ReLU}}}
\def\gammamon{\gamma_{\mathsf{mon}}}

\def\mycomment#1{{\color{gray}\Comment{#1}}}

\def\sumword{\mathsf{sum}}

\def\combineU{\mathord{\mathsf{combine}_\cup}}
\def\combineX{\mathord{\mathsf{combine}_\times}}

\newenvironment{proofsketch}{%
  \renewcommand{\proofname}{Proof (sketch)}\proof}{\endproof}

\def\AggCQ{AggCQ\xspace}
\def\AggCQs{AggCQs\xspace}

\def\rel#1{\text{\textsc{#1}}}
\def\att#1{\mathsf{#1}}

\def\qedexample{\hfill$\Diamond$}

\title{Tractability Frontiers of the Shapley Value for Aggregate Conjunctive Queries}

\AtEndPreamble{%
\theoremstyle{acmdefinition}
\newtheorem{remark}[theorem]{Remark}
\newtheorem{observation}[theorem]{Observation}}
\AtBeginDocument{%
  }

\begin{document}

\setcopyright{acmcopyright}
\acmJournal{PACMMOD}
\acmYear{2025} \acmVolume{3} \acmNumber{5 (PODS)}
\acmArticle{284} \acmMonth{11} \acmPrice{}
\acmDOI{10.1145/XXXXXXX}

\acmArticle{0}

\author{Christoph Standke}
\orcid{0000-0002-3034-730X}
\authornotemark[1]
\email{bennyk@cs.technion.ac.il}
\affiliation{%
  \institution{RWTH Aachen University}
  \city{Aachen}
  \country{Germany}
}

\author{Benny Kimelfeld}
\orcid{0000-0002-7156-1572}
\authornotemark[1]
\email{bennyk@cs.technion.ac.il}
\affiliation{%
  \institution{Technion}
  \city{Haifa}
  \country{Israel}
}
\affiliation{%
  \institution{RelationalAI}
  \country{USA}
}



\begin{abstract}
In recent years, the Shapley value has emerged as a general game-theoretic measure for assessing the contribution of a tuple to the result of a database query.
We study the complexity of calculating the Shapley value of a tuple for an aggregate conjunctive query, which applies an aggregation function to the result of a conjunctive query (CQ)  based on a value function that assigns a number to each query answer. Prior work by Livshits, Bertossi, Kimelfeld, and Sebag (2020) established that this task is \#P-hard for every nontrivial aggregation function when the query is non-hierarchical with respect to its existential variables, assuming the absence of self-joins. They further showed that this condition precisely characterizes the class of intractable CQs when the aggregate function is sum or count. In addition, they posed as open problems the complexity of other common aggregate functions such as min, max, count-distinct, average, and quantile (including median). 

Towards the resolution of these problems, we identify for each aggregate function a class of hierarchical CQs where the Shapley value is tractable with every value function, as long as it is local (i.e., determined by the tuples of one relation). We further show that each such class is maximal: for every CQ outside of this class, there is a local (easy-to-compute) value function that makes the Shapley value \#P-hard. 
Interestingly, our results reveal that each aggregate function corresponds to a different generalization of the class of hierarchical CQs from Boolean to non-Boolean queries. In particular, max, min, and count-distinct match the class of CQs that are all-hierarchical (i.e., hierarchical with respect to all variables), and average and quantile match the narrower class of q-hierarchical CQs introduced by Berkholz, Keppeler, and Schweikardt (2017) in the context of the fine-grained complexity of query answering. Finally, we show an example of an aggregate function---``has-duplicates''---that gives rise to a new variation of hierarchy that is more restricted than q-hierarchy.
\end{abstract}

\maketitle

\section{Introduction}
Significant research efforts have been dedicated in recent decades to devising concepts, formalisms, and methodologies for explaining the outcomes of algorithms in data science, particularly in machine learning and databases~\cite{DBLP:journals/ftdb/GlavicMR21,DBLP:conf/sigmod/PradhanLGS22,DBLP:conf/nlpcc/XuUDF0Z19,DBLP:journals/jdiq/BertossiG20,DBLP:journals/inffus/ChouMBOJ22}. Within this broader context, various approaches have been proposed and studied for quantifying the contribution of a database tuple to the result of a query~\cite{DBLP:journals/pvldb/MeliouGMS11,DBLP:conf/tapp/SalimiBSB16,DBLP:journals/lmcs/LivshitsBKS21,DBLP:journals/pvldb/MeliouRS14}. Our work focuses on the game-theoretic approach of the \e{Shapley value}~\cite{Shapley}, a well-established function for distributing wealth in cooperative games, underpinned by various theoretical justifications~\cite{roth,Shapley}. The Shapley value has found broad application across diverse fields such as economics, law, environmental science, and network analysis, as well as in data-centric paradigms including knowledge representation~\cite{DBLP:journals/ai/HunterK10}, machine learning~\cite{lundberg2017unified,lundberg2020local,DBLP:conf/ijcai/RozemberczkiWBY22}, and databases~\cite{DBLP:journals/sigmod/BertossiKLM23,DBLP:journals/lmcs/LivshitsK22,DBLP:conf/sigmod/DeutchFKM22,DBLP:journals/pacmmod/BienvenuFL24,DBLP:journals/sigmod/LivshitsBKS21,DBLP:conf/sigmod/LuoP24}. Specifically, prior work by Livshits, Bertossi, Kimelfeld, and Sebag~\cite{DBLP:journals/lmcs/LivshitsBKS21} initiated the study of the complexity of calculating the Shapley value of a database tuple's contribution to the result of a conjunctive query (CQ). This problem has subsequently been extended to related game-theoretic measures, such as the Banzhaf Power Index~\cite{DBLP:journals/pacmmod/AbramovichDF0O24} and, more generally, \e{Shapley-like} scores~\cite{DBLP:journals/pacmmod/KarmakarMSB24}.

In this manuscript, we investigate the complexity of computing the Shapley value for database tuples contributing to the numeric result of an \e{aggregate conjunctive query} (AggCQ). An AggCQ is defined as $\alpha\circ\tau\circ Q$, evaluated by first applying a conjunctive query $Q$, then transforming each resulting tuple into a numerical value via a \e{value function} $\tau$, and finally aggregating this bag of numbers into a single numerical result using an aggregate function $\alpha$. Throughout this section, we implicitly assume that CQs are without self-joins. Furthermore, our primary focus is on \e{localized} value functions $\tau$, 
which are determined by the attributes of a single relation. 
We defer a discussion on non-localized value functions to \Cref{sec:finer-beyond-localized}.
Moreover, we always assume that $\tau$ is computable in polynomial time.

The complexity analysis by Livshits et al.~\cite{DBLP:journals/lmcs/LivshitsBKS21}, along with subsequent work~\cite{DBLP:journals/pacmmod/BienvenuFL24,DBLP:conf/pods/ReshefKL20}, primarily focused on the Shapley value for Boolean CQs and, more generally, for a database tuple's contribution to an answer tuple in the output. We refer to this specific task as \e{membership}. They established that the tractability criterion for CQs in this context is the property of being \e{\exhierarchical}, meaning hierarchical\footnote{Recall that "hierarchical" implies that for any two variables, their corresponding sets of atoms are either disjoint or one contains the other~\cite{DBLP:conf/vldb/DalviS04}.} with respect to the existential variables. This criterion precisely mirrors that for CQ evaluation over tuple-independent probabilistic databases~\cite{DBLP:conf/vldb/DalviS04,DBLP:journals/vldb/DalviS07}. The similarity in complexities is not coincidental, as computational hardness for both problems boils down to \e{model counting}. Bienvenu, Figueira, and Lafourcade~\cite{DBLP:journals/pacmmod/BienvenuFL24} 
and Kara, Olteanu, and Suciu~\cite{DBLP:journals/pacmmod/0002OS24}
further explored the in-depth connection between the Shapley value and model counting. Livshits et al.~\cite{DBLP:journals/lmcs/LivshitsBKS21} initially motivated the study of membership, in part, due to the potential for extending this task to aggregate queries. They demonstrated this potential with two key findings: on the negative side, they showed that being \exhierarchical is necessary for the tractability of the Shapley value for every non-trivial aggregate function (under conventional complexity assumptions); on the positive side, they proved it is also sufficient for the aggregate functions sum and count.

The algorithm of sum (and count) over \exhierarchical CQs leverages the linearity of expectation. Livshits et al.~\cite{DBLP:journals/lmcs/LivshitsBKS21} posed as open problems the cases of other common aggregate functions that do not seem to benefit from the linearity of expectation, specifically the functions min, max, average, and quantile.  
As a preliminary step, they gave a polynomial-time algorithm for min and max in the special case of a single-atom CQ.
Recent work \cite{DBLP:journals/corr/abs-2506-16923} that was done independently of this paper extended the algorithm for min and max to the class of \e{\allhierarchical} CQs, that is, CQs hierarchical w.r.t.~\e{all} variables (and not just the existential ones as in \exhierarchical).
Yet, prior work left unanswered the question of tractability frontiers for other common aggregate functions and, in particular, whether the tractability frontier of sum and count also applies to these.

In this manuscript, we answer this question (negatively) for min, max, average, and quantile, as well as count-distinct (and also ``has-duplicates'' that we refer to later on) in the case of AggCQs  and with localized value functions. In particular, we identify for each aggregate function a class of hierarchical CQs where the Shapley value is tractable with \e{every} localized value function. Importantly, we show that each such class is \e{the maximal possible}, forming a frontier, in the following sense: for every CQ  outside of this class, there is a local value function where the Shapley value of a fact is \#P-hard.
We elaborate on our results in the remainder of this section.

\paragraph{Contributions.}

\begin{figure}
  \def\cit{\cite{DBLP:journals/lmcs/LivshitsBKS21}}
  \def\aggn#1{{\color{MidnightBlue}#1}}
  \input{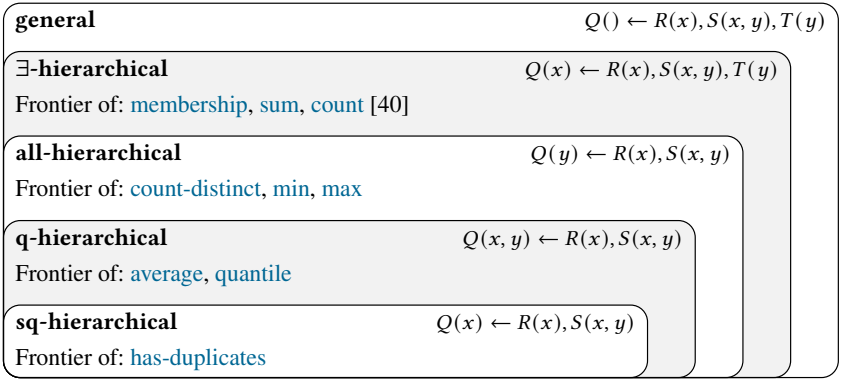}
\caption{Containment among classes CQs without self-joins. For each class, the box states the aggregate functions where the class captures precisely the tractable CQs without self-joins for Shapley computation.
\label{fig:hierarchy}}
\Description{Containment among classes CQs without self-joins. For each class, the box states the aggregate functions where the class captures precisely the tractable CQs without self-joins for Shapley computation.}
\end{figure}

We begin by showing (in \Cref{prop:constant-same-as-boolean}) a property of aggregate functions that, when held, makes the Shapley value intractable for every AggCQ, unless it is \allhierarchical. This property holds for all aforementioned aggregate functions, except for the functions sum and count; hence, we establish a negative answer to the above open question.  In the case of min, max, and count-distinct, the property of being \allhierarchical is not only necessary for tractability, but also sufficient. Hence, we establish a full classification for these three functions.

Nevertheless, we show that being \allhierarchical is insufficient for the aggregate functions \e{average} and \e{quantile}: there are \#P-hard average and quantile AggCQs with \allhierarchical CQs. We establish that the tractable class of CQs for these two aggregate functions is precisely that of the \e{\qhierarchical} CQs, introduced by Berkholz, Keppeler and Schweikardt~\cite{DBLP:conf/pods/BerkholzKS17} in the context of structure maintenance for answering CQs with constant-time delay. To the best of our knowledge, the property of \qhierarchical has been shown so far to characterize tractable fragments only in the scope of fine-grained complexity (separating different classes of polynomial time, e.g.,~\cite{DBLP:conf/pods/BerkholzKS17,DBLP:conf/icdt/MunozRV24,DBLP:conf/sigmod/IdrisUV17,DBLP:conf/pods/BerkholzM21,DBLP:journals/pvldb/WangHDY23}). Hence, this work is the first to show that \qhierarchical also characterizes tractability in a natural context of coarse-grained complexity (separating polynomial time from exponential time). The hardness proofs are via matrix inversion, as done previously for proving the hardness of Shapley-value calculation~\cite{DBLP:journals/lmcs/LivshitsBKS21,DBLP:conf/stacs/AzizK14,DBLP:journals/lmcs/LivshitsK22,DBLP:journals/pacmmod/0002OS24,DBLP:journals/pacmmod/KarmakarMSB24}; notably, in the case of average, our matrix construction uses the Kronecker product of the Hilbert matrix and the Hankel matrix with factorial entries. 

Recall that we started with the class of \exhierarchical CQs for the tractability of sum and count, restricted it to the \allhierarchical CQs for min, max and count distinct, and needed to further restrict it to the \qhierarchical CQs for average and quantile.
 Is the class of \qhierarchical CQs the ``ultimate extension'' of a single relation into multi-way joins?
More precisely, is it the case that,
for every aggregate function $\alpha$ (possibly with certain %
natural assumptions), if the Shapley value can be computed in polynomial time over a single relation, then it can also be computed in polynomial time on every \qhierarchical CQ? The answer turns out to be false.

We identify a natural aggregate function where the Shapley value can be computed in polynomial time over a single relation, yet is intractable on certain \qhierarchical CQs. This is the aggregate function ``has-duplicates'' that tests whether its input list of values includes at least one duplicate. (This aggregate operation is non-conventional in traditional database systems since it can be implemented easily via grouping and the count aggregation; yet, it is supported in the \textsf{pandas} library via the property \textsf{has\_duplicates} of the \textsf{Index} object of a \textsf{DataFrame}.) Moreover, we delineate the precise class of CQs where this aggregate function is tractable---this class is defined by restricting the class of \qhierarchical CQs so that no free variable can have a set of atoms that is strictly contained in that of any other atom. We call CQs of this class \e{strongly \qhierarchical}, or \e{\sqhierarchical} for short. 
This class, like all other classes mentioned earlier, coincides with the class of hierarchical CQs in the Boolean case.
Whether this restricted class is the ``ultimate extension'' in the sense described above remains an open problem for future investigation.

\Cref{fig:hierarchy} summarizes the results of this manuscript. Each box represents a class of CQs without self-joins, and their layout reflects containment among the classes, with ``general'' being the most general and ``\sqhierarchical'' being the most restrictive. In each box, we list the aggregate functions $\alpha$ for which the corresponding CQ class is a \emph{tractability frontier}; that is, it is the precise class of CQs without self-joins for which \AggCQs with $\alpha$ admit a polynomial-time computation of the Shapley value. For each box, we also provide an example of a CQ that belongs to the class but not to the more restrictive class that directly contains it.
The class of \exhierarchical CQs consists of those CQs for which the Shapley value of a database tuple, with respect to a given answer tuple in the result of the CQ, can be computed in polynomial time (and we refer to this result as \e{membership} in \Cref{fig:hierarchy}); this class is also the tractability frontier for the sum and count aggregate functions~\cite{DBLP:journals/lmcs/LivshitsBKS21}. For count-distinct, min, and max, the tractability frontier is the more restricted class of \allhierarchical CQs. Strictly harder are the cases of average and quantile (for every threshold), where the tractability frontier is the class of \qhierarchical CQs. The aggregate function has-duplicates is the hardest we consider, as its tractability frontier is the most restricted---the class of \sqhierarchical CQs.

The intuition behind the separation is as follows. The \exhierarchical property enables the decomposition of the database into smaller sub-databases, each corresponding to a distinct answer. These sub-databases may overlap, which is acceptable for the sum and count aggregates due to the linearity of expectation~\cite{DBLP:journals/lmcs/LivshitsBKS21}. However, this is insufficient for other aggregate functions. In contrast, the \allhierarchical property allows us to apply the standard algorithm for hierarchical CQs~\cite{DBLP:conf/vldb/DalviS04,DBLP:journals/vldb/DalviS07} that (conceptually) operates by repeatedly partitioning the database into disjoint subsets. This approach is suitable for min and max (though a different argument is used for count-distinct). Nonetheless, these disjoint subsets can still yield overlapping sets of answers, which poses a problem to the Shapley computation for the average and quantile aggregates. The \qhierarchical property addresses this by ensuring that the answer sets are also disjoint. For the aggregation has-duplicates, even this property is insufficient, since the Shapley computation requires not only disjoint answer sets but also disjoint sets of \e{values} extracted from the answers for aggregation. This stricter requirement can be accommodated when using \sqhierarchical CQs.

\paragraph{Organization.}
The rest of the manuscript is organized as follows. We first introduce the general framework, including the problem definition (\Cref{sec:framework}). 
We then establish general machinery that applies to all (or most) of the aggregate functions we study in the manuscript, including hardness under a general condition and an algorithmic template for hierarchical CQs (\Cref{sec:machinery}). We study the aggregate functions count-distinct, min, and max (\Cref{sec:all}), the aggregate functions average and quantile (\Cref{sec:q}), and then the function has-duplicates (\Cref{sec:duplicates}). 
Finally, we address the standing challenge of extending the classification results of this paper to incorporate the variability of the value function
(\Cref{sec:finer}) and conclude (\Cref{sec:conclusions}).
Due to space constraints, some of the proofs are omitted from the body of the manuscript and are given in the Appendix.

\section{Preliminaries}\label{sec:framework}
We begin by introducing the formal framework of the manuscript, including the problem definition.

\paragraph{Sets and bags.}
Let $X$ be a set and $k$ be a number. A \e{$k$-subset} of $X$ is a set $Y\subseteq X$ with $|Y|=k$. By $\binom{X}{k}$ we denote the set of all $k$-subsets of $X$.
Recall that a \e{bag} is a pair $(X,\mu)$ where $X$ is a set and $\mu:X\ra\naturals$ is a function that maps every element $x\in X$ to a \e{multiplicity}. For a bag $B = (X,\mu)$ we write $x\in B$ for $\mu(x) > 0$. %
The size of a bag $B=(X,\mu)$ is $\size{B} = \sum_{x \in X}\mu(x)$. We denote by $\finbags(X)$  the set of all finite bags over $X$.
If $f:X\ra Y$ is a function, then we use the conventional notation of $\multiset{f(x)\mid x\in X}$ to denote the bag $(Y,\mu)$ where $\mu(y)=|\set{x\in X\mid f(x)=y}|$, that is, the bag that is obtained from $X$ by replacing each element $x$ with the element $f(x)$ while keeping duplicates.

\paragraph{The Shapley value.} A \e{cooperative game} is a pair $(P,\nu)$ where $P$ is a set of players and $\nu:2^P\rightarrow\reals$ is a \e{utility function} that associates every coalition $C\subseteq P$ with a value $\nu(C)$, so that $\nu(\emptyset)=0$. 

Consider a situation when we gather the players of $P$ by selecting players iteratively uniformly at random without replacement, starting with the empty set. The Shapley value of a player $p\in P$ is the expected increase in utility when adding $p$~\cite{Shapley,roth}. Using the law of conditional expectation via conditioning over the set of players chosen before $p$, it can be computed by the following formula.%
    \begin{equation}\label{eq:shapley}
    \Shapley(p, \nu) 
    = \sum_{C \subseteq P \setminus\set{p}} q_{\size{C}} \cdot 
    (\nu(C \cup \set{p}) - \nu(C))
    \end{equation}
where $q_{\size{C}} \defeq \frac{|C|!(|P|-|C|-1)!}{|P|!}$ is the probability that the iterative process selects $C$ first, then $p$, and then the remaining $|P|-|C|-1$ players.

\paragraph{Relational databases.}
A \e{relation schema} has the form $R/k$, where $R$ is a \e{relation name} and $k\in\naturals$ is an \e{arity}. We assume an infinite domain $\consts$ of \e{constants} that occur in database tuples. A \e{relation} over the relation schema $R/k$ is a finite subset of $\consts^k$. A \e{database schema} (or just \e{schema} for short) is a finite set $\scs$ of relation schemas with distinct names. We denote by $\arity(R)$ the arity $k$ associated with the relation name $R$ (assuming that $\scs$ is known from the context). A \e{database} $D$ over the schema $\scs$ assigns to each relation schema $R/k$ a relation over $R/k$, and we denote this relation by $R^D$.
A \e{fact} over $\scs$ is an expression of the form $R(a_1,\dots,a_k)$ where 
$R/k \in \scs$ and each $a_i$ is in $\consts$.  We identify a database $D$ over a schema $\scs$ as the set of its facts over $\scs$. For example, $D_1\subseteq D_2$ means that each relation of $D_2$ contains the corresponding relation of $D_1$.

\paragraph{Conjunctive queries.}
Let $\scs$ be a schema. A \e{relational query} $Q$ is associated with an arity, denoted $\arity(Q)$, and it is abstractly a function that maps a given database $D$ over $\scs$ into a finite $k$-ary relation $Q(D)\subseteq\consts^k$, where $k=\arity(Q)$. 

We will focus on the special case where the query is a \e{conjunctive query} (CQ), which is an expression of the form
$$Q(\vec x) \la R_1(\vec z_1),\dots,R_q(\vec z_q)$$
where $\vec x$ is a sequence of variables, each $R_j$ is the name of a relation in $\scs$, and each $\vec z_j$ is a sequence of length $\arity(R_j)$ consisting of variables and constants.
We denote by $\vars(Q)$ the set of all variables of $Q$ (i.e., those that occur in either $\vec x$ or some $\vec z_j$), by $\freevars(Q)$ the set of \e{free} variables of $Q$, that is, those in $\vec x$, and by $\exvars(Q)$ the set of \e{existential} variables of $Q$, that is, those in $\vars(Q)$ that are not in $\vec x$. Hence, $\vars(Q)=\freevars(Q)\cup\exvars(Q)$. In addition, we denote by $\atoms(Q)$ the set of atomic formulas $R_j(\vec z_j)$ of $Q$, and by $\atoms(Q,x)$ the set of atoms of $Q$ where the variable $x$ occurs. By a \e{self-join} we refer to a repeated relation symbol in the body of $Q$. Throughout the manuscript, we always assume that the involved CQs have no self-joins.

We adopt the usual semantics of CQs, defined as follows. A \e{homomorphism} $h$ from the CQ $Q$ to a database $D$ is a function $h:\vars(Q)\rightarrow\consts$ that maps every variable $y$ of $Q$ to a constant $h(y)$, so that for every atom $R(\vec z)$ of $Q$, the fact $R(h(\vec z))$ is in $D$; here, $h(\vec z)$ denotes the sequence obtained from $\vec z$ by replacing every variable $y$ with the constant $h(y)$.
We denote by $\hom(Q,D)$ the set of all homomorphisms from $Q$ to $D$. 
Then $Q(D)\defeq \set{h(\vec x)\mid h\in\hom(Q,D)}$. 

\paragraph{Hierarchical CQs.}
We refine the well-known concept of a \e{hierarchical} CQ, as follows. Let $Q$ be a CQ. We say that $Q$ is \e{hierarchical} with respect to (w.r.t.) a set $V$ of variables if, for all variables $x$ and $y$ in $V$, it holds that (a) 
$\atoms(Q,x)\subseteq\atoms(Q,y)$, or (b) 
$\atoms(Q,y)\subseteq\atoms(Q,x)$, or (c)
$\atoms(Q,x)$ and $\atoms(Q,y)$ are disjoint. We say that $Q$ is \e{\exhierarchical} if $Q$ is hierarchical w.r.t.~the set $\exvars(Q)$ of existential variables, and that $Q$ is \e{\allhierarchical} if $Q$ is hierarchical w.r.t.~the entire set $\vars(Q)$ of variables of $Q$. The CQ $Q$ is \e{\qhierarchical} if it is \allhierarchical and for all variables $x$ and $y$, if $\atoms(Q,y)\subsetneq\atoms(Q,x)$ and $y$ is free, $x$ is free as well; in other words, there are no existential $x$ and free $y$ such that $\atoms(Q,y)\subsetneq\atoms(Q,x)$.

Note that an \allhierarchical CQ is also \exhierarchical, but not necessarily vice versa, and that every \qhierarchical CQ is  \allhierarchical. Hence, we have the following entailment chain:
$
\mbox{ \qhierarchical }\ra
\mbox{ \allhierarchical }\ra
\mbox{ \exhierarchical }
$.
\Cref{fig:hierarchy} gives examples of CQs of the different classes. For every class (represented by a box), its CQ is not in any class contained in it.

\begin{remark}
All three classes of hierarchical CQs coincide when the CQ is Boolean---the distinction is meaningful only in the presence of free variables, which are essential for the aggregation. The concept of a hierarchical CQ was originally coined by Dalvi and Suciu~\cite{DBLP:conf/vldb/DalviS04}. The relaxation to \exhierarchical was adopted later on (under different terminology) by Fink and Olteanu~\cite{DBLP:journals/tods/FinkO16}. The restriction to q-hierarchical CQs was introduced soon after by 
Berkholz, Keppeler and Schweikardt~\cite{DBLP:conf/pods/BerkholzKS17} in the context of fine-grained complexity of query maintenance through database updates.\qedexample
\end{remark}

\paragraph{Aggregate queries.}
Let $\scs$ be a schema. In this work, an \e{aggregate query} $A$ consists of three components:
\e{(1)} a relational query $Q$;
\e{(2)} a \e{value function} $\tau:\consts^{k}\rightarrow\reals$ where $k=\arity(Q)$;
\e{and (3)} an \e{aggregation function} $\alpha:\finbags(\reals)\rightarrow\reals$. We will make the convenient assumption that $\alpha$ is zero on the empty set, that is, $\alpha(\emptyset)=0$.
The semantics of $A=\alpha\circ\tau\circ Q$ is fairly straightforward: given a database $D$, the result $A(D)$ of applying $A$ to $D$ is the number
$$\alpha(\multiset{\tau(\vec t) \mid \vec t\in Q(D)})\,.$$

We will discuss general value functions $\tau$ and make the assumption that $\tau(\vec t)$ is computable in polynomial time. Another important assumption we make in most of our results is that $\tau$ is \e{localized} in the sense that it is determined by one of the atoms in the CQ. 
Formally, we say that $\tau$ is \e{localized} on an atom $R(\vec z)$ of $Q(\vec x)$ if for all databases $D_1$ and $D_2$ and homomorphisms $h_1\in\hom(Q,D_1)$ and $h_2\in\hom(Q,D_2)$, if $h_1(\vec z)=h_2(\vec z)$, then $\tau(h_1(\vec x))=\tau(h_2(\vec x))$; that is, $\tau$ is completely determined by the assignment to the variables of $R$. In this case, we may write $\tau(h(\vec z))$ instead of $\tau(h(\vec x))$.
Furthermore, we write $\tau \equiv c$ if $\tau$ is the constant function mapping each tuple in $\consts^{\arity(Q)}$ to the value $c \in \reals$.

Throughout the manuscript, we will refer to the following specific value functions, assuming that the relevant $\consts$ value is a number:
\begin{align}
\label{eq:tauidi}
\tauidi{i}(a_1,\dots,a_k)
& \,\defeq\, a_i \\
\label{eq:taugbi}
\taugbi{b}{i}(a_1,\dots,a_k)
& \,\defeq\,
\mbox{$1$ if $a_i>b$, and $0$ otherwise}
\\
\label{eq:taurelu}
\taurelui{i}(a_1,\dots,a_k)
& \,\defeq\,
\mbox{$a_i$ if $a_i>0$, and $0$ otherwise}
\end{align}
Here, the parameter $b$ can be any fixed number.
Note that each of the three is localized on some atom, since it is determined by one entry of the tuple.
In the case where $k=1$, we may omit the superscript and write just $\tauid$, $\taugb{b}$, and $\taurelu$ (where ``ReLU'' stands for Rectified Linear Unit).

For $\alpha$, we will focus on several common functions: sum ($\aggsum$), count-distinct ($\aggcdis$), average ($\aggavg$), minimum ($\aggmin$), maximum ($\aggmax$), median ($\aggmed$), and more generally $q$-quantile ($\aggquantile_q$). As usual, we set $\aggquantile_q(B) = \frac{1}{2}(x_{\lceil q\size{B}\rceil}  + x_{\lfloor q\size{B} + 1\rfloor})$ where $x_i$ is the $i$-th largest value of $B$. For our exploration, we will also explore the aggregate function \e{has-duplicates} ($\aggduplicates$) that evaluates to $1$ if the multiplicity of at least one element is two or more, and $0$ otherwise.
In notation:
\[
\text{For $B=(X,\mu)$, we have:\;}\aggduplicates(B) = \begin{cases}
1 & \text{if there is } a\in X \text{ with } \mu(a) \geq 2 \\
0 & \text{otherwise.}
\end{cases}
\]

\begin{example}\label{example:earns}
The following schema $\scs$ is used for the database of an educational institute offering individual courses:  
$
\rel{Earns}(\att{person},\att{salary})
,\;
\rel{Course}(\att{name},\att{number})
,\;
\rel{Took}(\att{person},\att{course})
$.
The following \AggCQ returns the average salary of people who took courses in the institution.
\begin{equation}\label{eq:earns}
A\;=\;
\aggavg\circ s\circ \big(Q(p,s)\leftarrow
\rel{Earns}(p,s),
\rel{Took}(p,c),
\rel{Course}(\_,c)
\big)
\end{equation}
Here, $\alpha=\aggavg$ and $\tau(p,s)$ is simply $s$; in particular, $\tau$ is localized (on the relation $\rel{Earns}$).
Note that a person may take many courses, but she is counted only once for the average due to the projection of $Q$. An example of a non-localized $\tau$ is in the following \AggCQ over the schema with
$\rel{Cargo}(\att{number},\att{weight})$,
$\rel{Truck}(\att{number},\att{weight})$, and
$\rel{Carries}(\att{truck},\att{cargo})$.
\begin{equation*}
A'\;=\;
\aggmax\circ (w_c+w_t)\circ \big(Q(c,t,w_c,w_t)\la
\rel{Cargo}(c,w_c),
\rel{Carries}(t,c),
\rel{Truck}(t,w_t)
\big)
\end{equation*}
This query asks for the maximal weight of a truck loaded with cargo. The value function $\tau$ is non-localized since its determination requires attributes from both $\rel{Cargo}$ and $\rel{Truck}$.
\qedexample
\end{example}

\paragraph{Shapley value for \AggCQs.}
We study the computational complexity of measuring the contribution of a fact $f$ to the result of an aggregate query $A$ over a database $D$. We model this contribution as the Shapley value of $f$ in the cooperative game $(P,\nu)$, where $P$ is the database $D$ (which is a set of facts, each being a player) and the utility $\nu$ is defined by $\nu(C)\defeq A(C)$ for every $C\subseteq P$. Moreover, following the common modelling of the responsibility for database queries~\cite{DBLP:journals/pvldb/MeliouGMS11,DBLP:journals/lmcs/LivshitsBKS21}, we assume that $D$ consists of \e{endogenous} facts and \e{exogenous} facts. Facts of the latter kind are taken for granted and do not participate in the game; hence, the set $P$ of players is restricted to the endogenous facts.

Formally, we assume that a given database $D$ is the union of two disjoint databases (over the same schema $\scs$): the database $D\endo$ contains the endogenous facts, and the database $D\exo$ contains the exogenous facts. For a fixed schema $\scs$ and aggregate query $A$, the computational problem is defined as follows. Given a database $D$ and a fact $f\in D\endo$, calculate the value $\Shapley(f, A)[D]$ as defined in \Cref{eq:shapley}, where:
\begin{itemize}
    \item The set $P$ of players is $D\endo$.
    \item The utility function $\nu$ is defined by $\nu(C)\defeq A(C\cup D\exo)-A(D\exo)$.
\end{itemize}
To simplify notation, we drop $D$ when it is clear from context.

\begin{example}
Consider again the \AggCQ of \Cref{eq:earns}. Suppose that we wish to measure the contribution of every course to the average salary (or the median if $\alpha$ is $\frac12$-$\aggquantile$) of the past graduates. Then we make all $\rel{Course}$ facts endogenous and all other facts exogenous. Given a specific course represented as a fact $f=\rel{Course}(c,n)$, its contribution would be quantified as $\Shapley(f, A)$. \qedexample
\end{example}

\begin{remark}\label{remark:rationality}
Previous work in database management has adopted the Shapley value largely because it is a conventional measure across disciplines, including other areas of data science~\cite{DBLP:conf/ijcai/RozemberczkiWBY22,DBLP:conf/aistats/JiaDWHHGLZSS19,DBLP:journals/inffus/ChouMBOJ22}. This makes it an attractive candidate for a well-defined, application-independent attribution measure that can be offered as part of a database system service.
In certain situations, the axiomatic properties that characterize the Shapley value are indeed compelling. 

For example, suppose that endogenous tuples represent employees or teams, and attribution serves to determine their reward for success. Further suppose that success is defined by an aggregate query $Q$ that can be expressed as the sum of other aggregate queries $Q'$ (e.g., the number of distinct new customers in the US equals the sum of the numbers of distinct new customers across the states). In such a setting, it is arguably desirable that the reward remain the same whether it is based on $Q$ directly or it is given for each $Q'$ separately (linearity). Similarly, it may appear fair that two facts indistinguishable in their impact on the query receive the same reward (symmetry), and that facts with no impact receive none (null player). Finally, to determine whether an attribution score is considered high or low, it is useful that the scores are calibrated so that the contributions of all endogenous tuples sum up to the actual result of the aggregate query (efficiency). 
We discuss further the rationality of the Shapley value for \AggCQs in \Cref{sec:conclusions}.
\qed
\end{remark}

We focus on the \e{data complexity} of computing $\Shapley(f, A)$. This means that the \AggCQ $A$ is assumed to be fixed, and the input consists of the database $D$ and the fact $f\in D$. When we prove \fpsharpp-completeness, we show only \fpsharpp-hardness; completeness is done in standard techniques for showing membership of probability computation in \fpsharpp~\cite{DBLP:conf/pods/GradelGH98,DBLP:journals/jacm/FaginKK11}.

\section{General Machinery}
\label{sec:machinery}
In this section, we present general machinery that we use throughout the manuscript, both for establishing hardness results on the Shapley value and for algorithms to compute this value.

\subsection{General Hardness Result}

We begin with a general hardness result. Our starting point is the following result of
Livshits et al.~\cite{DBLP:journals/lmcs/LivshitsBKS21}, showing the hardness of the Shapley value for aggregate queries whenever the CQ is non-\exhierarchical and without self-joins.

\begin{theorem}
\emph{\cite[Theorem 4.8]{DBLP:journals/lmcs/LivshitsBKS21}}\;
\label{thm:livshits-hardness}
Let $A=\alpha\circ\tau\circ Q$ be an \AggCQ. Suppose that $Q$ has no self-joins. If $Q$ is not \exhierarchical, then it is \fpsharpp-complete to compute
$\Shapley(f, A)$, unless $A$ is a constant query (namely, $0$, since we assume $\alpha(\emptyset) = 0$.).
\end{theorem}
From \Cref{thm:livshits-hardness}  we conclude
that, for every nontrivial \AggCQ, a necessary condition for polynomial-time Shapley computation is that the underlying CQ is \exhierarchical. The same work also showed that, for $\aggsum$ (and $\aggcount$), this condition is also sufficient. The next argument establishes that this is no longer true for the aggregate functions that we study in this work. 

An aggregate function $\alpha$ is \e{constant per singleton} if for all $a\in\reals$ and nonempty bags $B$ and $B'$ over the singleton $\set{a}$ it holds that $\alpha(B)=\alpha(B')$.

\begin{restatable}{proposition}
{propconstantsameasboolean}
\label{prop:constant-same-as-boolean}
Let $A = \, \alpha\circ\tau\circ Q$ be an aggregate query such that $\tau$ is a constant function $\tau \equiv c$ and $\alpha$ is constant per singleton. For all databases $D$ and facts $f\in D$ it holds that
$$\Shapley(f, A)=
\alpha(\multiset{c})\cdot
\Shapley(f, \Qbool)$$ where $\Qbool$ is the Boolean CQ obtained from $Q$ by viewing every free variable as existential.
\end{restatable}

Under the assumption that $\alpha(\multiset{c})$ is nonzero, \Cref{prop:constant-same-as-boolean} immediately yields a Turing reduction from computing $\Shapley(f,\Qbool)$ to computing $\Shapley(f,A)$ via division by $\alpha(\multiset{c})$. This shows that:

\begin{restatable}{theorem}{thmhardnessconstpersingle}
\label{thm:hardness-const-per-single}
Let $A=\alpha\circ\tau\circ Q$ be an \AggCQ such that $\tau$ is a constant function $\tau \equiv c$ and  $\alpha$ is constant per singleton with $\alpha(\multiset{c}) \neq 0$. Suppose that $Q$ has no self-joins. If $Q$ is not \allhierarchical, then it is \fpsharpp-complete to compute
$\Shapley(f, A)$. In particular, this is the case for the aggregate functions $\aggmin$, $\aggmax$, $\aggcdis$, $\aggavg$, and $\aggquantile_q$ for all $q \in (0,1)$.
\end{restatable}

Hence, for the \AggCQs we study in this work, the class of tractable underlying CQs is contained in the class of \allhierarchical CQs. In particular, this class is strictly narrower than the class of tractable CQs for sum (and count), namely that of the \exhierarchical CQs~\cite{DBLP:journals/lmcs/LivshitsK22}.

\subsection{Algorithmic Template for Hierarchical CQs}\label{sec:template}

We now describe a generic template that we apply for computing $\Shapley(f, A)$ when the underlying CQ $Q$ is \allhierarchical. 
We begin with a folklore technique for reducing the computation to an expectation computation. 
Our goal is to compute $\Shapley(f, A)$ on a database $D$. In the following equations, we denote by $F$ the database that is identical to $D$, except that $f$ is exogenous rather than endogenous. We also denote by $G$ the database $D\setminus{f}$. 
Note the following (referring to \Cref{eq:shapley}):
\begin{align*}
\Shapley(f, A)
&=  \sum\nolimits_{E \subseteq D\endo \setminus\set{f}} q_{\size{E}} \cdot 
    (A(E \cup \set{f} \cup D\exo) - A(E \cup D\exo)) \\
&= \sum\nolimits_{k=0}^{|D\endo|-1}q_k 
\left( 
\Big(
\sum\nolimits_{E\in\binom{F\endo}k}A(E \cup F\exo) 
\Big)
- 
\Big(
\sum\nolimits_{E\in\binom{G\endo}k}A(E \cup G\exo) 
\Big)
\right)
\end{align*}
Therefore, it suffices to be able to compute the value 
$\sum_{E\in\binom{D\endo}k}(A(D\exo\cup E))$, that is, the sum of the results of $A$ over all of the databases that are obtained from $D$ by removing all but $k$ endogenous facts. 
Most of the algorithms we devise will compute this value, which we denote simply by $\sumword_k(A,D)$:
\begin{equation}\label{eq:sum-A-k}
\sumword_k(A,D)\defeq \sum\nolimits_{E\in\binom{D\endo}k}A(D\exo\cup E)
\end{equation}

Notably, whenever we present an algorithm for the Shapley value by calculating $\sumword_k(A,D)$, it applies to the class of the \e{Shapley-like scores}~\cite{DBLP:journals/pacmmod/KarmakarMSB24} that includes the Shapley value, as well as other contribution measures such as the \e{Banzhaf score}~\cite{Banzhaf_1965_5380} that was studied in the context of membership in CQ answers, either knowingly~\cite{DBLP:journals/pacmmod/AbramovichDF0O24} or unawarely~\cite{DBLP:conf/tapp/SalimiBSB16}. Unfortunately, other algorithms (that rely, e.g., on the linearity of expectation) are not known to generalize to Shapley-like scores, and so is the case for the lower bounds -- specialized reductions should be constructed for these (and possibly algorithms for classes that are considered intractable in our Shapley setting).

In general, the computation may not produce $\sumword_k(A,D)$ directly, but rather a value or function $\P[Q,D]$ from which $\sumword_k(A,D)$ is easily computable. The computation of $\P[Q,D]$ is done via the dynamic programming of \Cref{fig:generic-alg}, which entails the computation of $\P[Q',D']$ for ``smaller'' $Q'$ and $D'$. This program follows the general structure of typical algorithms for Boolean hierarchical CQs, such as the ones used for computing the probability of $Q$ in a probabilistic database~\cite{DBLP:journals/vldb/DalviS07} and computing the Shapley value of $Q$~\cite{DBLP:journals/lmcs/LivshitsBKS21}.
Note, however, that we do not assume that $Q$ is Boolean.

\begin{figure}[t]
\small
\hrule
\begin{minipage}{\linewidth}
\renewcommand{\arraystretch}{1.1}
\begin{tabular}{ll}
\multicolumn{2}{l}{
\em Generic dynamic programming for \allhierarchical CQs}\\
\textbf{Input:} & Database $D$, \allhierarchical CQ $Q$, and $\P[Q',D']$ for all relevant smaller $Q'$ and $D'$ \\
\textbf{Goal:} & 
Compute $\P[Q,D]$ 
\end{tabular}
\end{minipage}
\hrule
\begin{algorithmic}[1]
\If{$Q$ has no variables}
\mycomment{Constant number of relevant facts, compute and terminate}
\State Compute $\P[Q,D]$ directly and \textbf{return}
\EndIf
\If{$Q$ has a root variable $x$}\label{alg:choice-of-root}
\State Let $\set{a_1,\dots,a_n}$ be the values that $x$ can take
\State $D_{1}\gets$ the subset of $D$ consistent with 
$x\mapsto a_1$

\State $D_{\mathsf{tmp}}\gets D_1$

\State $\P[Q,D_{\mathsf{tmp}}]\gets \P[Q_{x\mapsto a_1},D_1]$
\For{$i=2$ to $n$}
\State $D_i\gets$ the subset of $D$ consistent with 
$x\mapsto a_i$
\State $\P[Q,D_i]\gets\P[Q_{x\mapsto a_i},D_i]$
\State $\P[Q,D_{\mathsf{tmp}}\cup D_i]\gets$ 
{\color{RoyalBlue}
$\combineU\left(\P[Q,D_{\mathsf{tmp}}],\P[Q,D_i]\right)$
}
\mycomment{$\P[Q,D_i]=\P[Q_{x\mapsto a_i},D_i]$}
\State $D_{\mathsf{tmp}}\gets D_{\mathsf{tmp}}\cup D_i$
\EndFor
\mycomment{Last $D_{\mathsf{tmp}}\cup D_i$ is $D$, so $\P[Q,D]$ is computed}
\Else 
\mycomment{$Q$ is a cross product}
\State Suppose that $Q$ is a cross product of $Q_1$ and $Q_2$
\State Let $D_1$ and $D_2$ be the subsets of $D$ with the relations of $Q_1$ and $Q_2$, respectively
\State 
$\P[Q,D]\gets$
{\color{RoyalBlue}
$\combineX\left(\P[Q_1,D_1],\P[Q_2,D_2]\right)$
}
\EndIf
\end{algorithmic}
\hrule
\caption{Generic algorithmic scheme for hierarchical CQs\label{fig:generic-alg}}
\Description{Generic algorithmic for hierarchical CQs}
\end{figure}

In the algorithm, we use the following terminology. 
A \e{root} variable of the CQ $Q$ is a variable $x$  that appears in every atom. The values that $x$ \e{can take} are the values of $D$ that occur in every column that corresponds to a position where $x$ occurs in $Q$.
If $a$ is such a value, 
then a fact $f$ is consistent with $x\mapsto a$ if $f$ can be obtained from an atom $R_j(\vec z_j)$ of $Q$ by replacing $x$ with $a$ and replacing every other variable with arbitrary constants. 
The \e{subset of $D$ consistent with $x\mapsto a$} is the database that consists of all facts of $D$ that are consistent with some atom of $Q$. In addition, we denote by $Q_{x\mapsto a}$ that CQ that is obtained by $Q$ by replacing each body occurrence of $x$ with $a$; in addition, if $x$ is a head variable, then $x$ is removed from the head. Hence, $\vars(Q_{x\mapsto a})=\vars(Q)\setminus\set{x}$.

We say that the CQ $Q$ is a \e{cross product} of the CQs $Q_1$ and $Q_2$ if
$Q_1$ and $Q_2$ have disjoint sets of relation symbols and disjoint sets of variables,
the set of atoms of $Q$ is the union of the sets of atoms of $Q_1$ and $Q_2$,
and the set of head variables of $Q$ is the union of the set of head variables of $Q_1$ and $Q_2$.
As an example, the CQ $Q(x,y)\la R(x),S(y,z),T(z)$ is the cross product of the CQs 
$Q_1(x)\la R(x)$ and $Q_2(y)\la S(y,z),T(z)$.

The above notation covers all terminology used in the generic algorithm of \Cref{fig:generic-alg}, except for two subroutines that are task-specific and need to be devised for each specific $\P$:
\begin{itemize}
    \item The subroutine $\combineU(\P[Q,D'],\P[Q,D''])$ computes $\P[Q,D'\cup D'']$, assuming that
    $D'$ and $D''$ are disjoint and 
    $Q(D'\cup D'')=Q(D')\cup D(D'')$.
\item The subroutine $\combineX(\P[Q',D'],\P[Q'',D''])$ computes 
$\P[Q,D'\cup D'']$ assuming that $Q$ is the cross product of $Q'$ and $Q''$.
\end{itemize}

Importantly, we need to clarify how $\tau$ is defined on the generated CQs. Recall that $\tau$ is localized on a single relation $R$. 
When we break $Q$ as a cross product of $Q'$ and $Q''$, the relation $R$ belongs to just one of the two, say $Q'$. The semantics of $\P$ on $Q''$ will not relate to $\tau$. 
However, we assume that $\tau$ is defined on $Q'$, since $\tau$ is localized on $R$ and, by definition, it is a function of the assignment to the free variables of $R$.

In the case of $Q_{x\mapsto a}$, the query answers may lose a value that is needed for $\tau$ (namely, $x=a$). Nevertheless, since this value is always $a$ in the answers to $Q_{x\mapsto a}$, we can assume that $\tau$ is defined on the answers of $Q_{x\mapsto a}$ by implicitly replacing $\tau$ with $\tau'$ defined by $\tau'(\vec t)=\tau({\vec t}_a)$ where ${\vec t}_a$ is obtained from $\vec t$ by inserting $a$ to every position where $x$ is in the head of  $Q$.

In summary, the following describes a recipe to instantiate the generic algorithm of \Cref{fig:generic-alg}:

\begin{mainbox}{Instantiating the generic algorithm for an aggregate function $\alpha$:} 
Let $A=\alpha\circ\tau\circ Q$ be an aggregate query, let $D$ be an input database, and let $f$ be a given fact.%
\begin{enumerate}
\item Define the content of the data structure $\P$.
\item Show how $\sumword_k(A,D)$ is computable from $\P[Q,D]$ for all $k=0,\dots,|D\endo|$.
\item Devise an algorithm for implementing $\combineU$
\item Devise an algorithm for implementing  $\combineX$.
\end{enumerate}
\end{mainbox}

\section{Count Distinct, Min, Max}\label{sec:all}

We start our analysis by considering the aggregation functions $\alpha \in \set{\aggcdis, \aggmin,\aggmax}$. First, we observe that all of them are constant per singleton. By \Cref{thm:hardness-const-per-single}, there is a value function $\tau$ such that Shapley value computation of $\alpha \circ \tau \circ Q$ is intractable for all non-\allhierarchical CQs $Q$. 
The following theorem states that being non-\allhierarchical is not only sufficient for the existence for such $\tau$, but also necessary.
\begin{theorem}\label{thm:max-count-distinct-classification}
Let $Q$ be a CQ without self-joins and $\alpha\in\set{\aggmin,\aggmax,\aggcdis}$. Denote by $A_\tau$ the aggregate query 
for $\alpha\circ\tau\circ Q$.
\begin{enumerate}
    \item If $Q$ is \allhierarchical, then $\Shapley(f,A_\tau)$ is computable in polynomial time, given a database $D$ and a fact $f$, for every localized $\tau$.
    \item Otherwise, there is a localized $\tau$ such that computing $\Shapley(f, A_\tau)$, given a database $D$ and a fact $f$, is \fpsharpp-complete. 
\end{enumerate}
\end{theorem}

In the next subsections, we will present polynomial-time algorithms for \allhierarchical queries.

\subsection{Count Distinct}

To get an intuition for the Shapley values for $\aggcdis$-aggregation, we first state a closed formula for simplest case of a single relational query. In that case, the Shapley value matches the intuitive contribution of a fact: one over the number of all facts with the same $\tau$-value.

\begin{restatable}{proposition}
{propclosedformulacountdisctinct}
Let $Q(\vec x) \la R(\vec x)$ and assume that all facts in $R$ are endogenous. Then
\[\Shapley(R(\vec t), \aggcdis \circ \tau \circ Q) = \frac{1}{\size{\set{R(\vec t') \, \mid \, \tau(\vec t') = \tau(\vec t)}}}.\]
\end{restatable}

Now, we present an algorithm for general \allhierarchical CQs that uses a direct reduction to the Boolean case. It is based on the insight that we can express $\aggcdis$ as a sum over indicator variables $\chi_a$ with $\chi_a(B) = 1$ if $a \in B$ and $0$ otherwise. Since the indicator variables can be computed using Boolean queries, we obtain the following lemma.

\begin{restatable}{lemma}
{lemreductioncdisboolean}
\label{lem:reduction-cdis-boolean}
Let $Q$ be a CQ without self-joins, $\tau$ be localized on the atom $R(\vec z)$, and $D$ be a database. For $a \in \reals$, let $D_a$ be the database obtained from $D$ by removing all facts $R(\vec b)$ with $\tau(\vec b) \neq a$. Furthermore, let $\Qbool \la Q(\vec x)$. Let $A = \aggcdis \circ \tau \circ Q$. Then, for each fact $f \in D\endo$, we have
    \[
    \Shapley(f,A)[D] = \sum\nolimits_{a \in (\tau \circ Q)(D)} \Shapley(f,\Qbool)[D_a],
    \]
    with the convention that $\Shapley(f,\Qbool)[D_a] = 0$ for $f\notin D_a$.
\end{restatable}

This lemma yields a polynomial-time algorithm and shows \Cref{thm:max-count-distinct-classification} for $\alpha = \aggcdis$: If $Q$ is \allhierarchical, then $\Qbool$ is a Boolean hierarchical CQ and we can use the polynomial-time algorithm from \cite{DBLP:journals/lmcs/LivshitsBKS21} to compute $\Shapley(f,\Qbool)[D_a]$ for the at most $\size{D}^{\arity(Q)}$ values $a \in (\tau \circ Q)(D)$.

\subsection{Min and Max}

Next, we present a polynomial-time algorithm for computing the Shapley value for $\aggmax$-aggregation on \allhierarchical queries.  For $\aggmin$-aggregation, we can use the same algorithm by first negating all values in the database $D$ and then negating all obtained Shapley values.
For intuition, we again start with a closed formula for the simplest case of a single-atom query. 
\begin{restatable}{proposition}
{propclosedformulamax}
Let $Q(\vec x) \la R(\vec x)$. Assume all facts are endogenous. Let $q_k$ as in \Cref{eq:shapley} and for $a \in \reals$, let $\mult{\leq}{a}$ and $\mult{<}{a}$ denote the number of facts $R(\vec t)$ with $\tau(\vec t) \leq a$ or $< a$, respectively. %
    \[\Shapley(R(\vec t), \aggmax \circ \tau \circ Q)
    = 
    \frac{\tau(\vec t)}{\size{D}} + \sum_{\substack{a \in (\tau \circ Q)(D) \\ a < \tau(\vec t)}}  (\tau(\vec t) - a)  \sum_{k = 1}^{n-1} q_k \Bigg( \binom{\mult{\leq}{a}}{k} - \binom{\mult{<}{a}}{k} \Bigg)\text{.}\]
\end{restatable}
The proposition illustrates, for instance, that even facts with small $\tau$-values contribute to the maximum aggregation.

Our algorithm for general \allhierarchical CQs is based on the generic algorithm shown in \Cref{fig:generic-alg}. 
Let $A = \aggmax \circ \tau \circ Q$ be such that $\tau$ is localized on an atom $R(\vec z)$.
We define the data structure $\P[Q',D']$ as follows:
\begin{itemize}
    \item If $D'$ contains the relation $R$ that determines $\tau$, then $\P[Q',D']$ is a map that maps every pair $(a,k)$ with $a \in (\tau \circ Q')(D')$ and $0 \leq k \leq \size{{D'}\endo}$ to the number of subsets $E \in \binom{{D'}\endo}{k}$ with $\max((\tau \circ Q')(E \cup {D'}\exo)) = a$. (Recall that $\tau$ is defined on $Q'$, as explained at the end of \Cref{sec:template}.)
    \item Otherwise, $\P[Q',D']$ maps each number $k = 0,1,\ldots,\size{{D'}\endo}$ to the number of subsets $E \in \binom{{D'}\endo}{k}$ such that $Q(E \cup {D'}\exo)$ is nonempty.
\end{itemize}

With this data structure at hand, we can simply compute $\sumword_k(A,D)$ 
via
\[
\sumword_k(A,D) = \sum_{E \in \binom{D\endo}{k}} (\aggmax \circ \tau \circ Q)(E \cup D\exo) = \sum_{a \in (\tau \circ Q)(D)} a \cdot \P[Q,D](a,k).
\]

We next describe the algorithms for $\combineU$ and $\combineX$ for this data structure. More details on the computation can be found in \Cref{apx:all}.

Let us first discuss the algorithm for  $\combineU(\P[Q', D_1], \P[Q', D_2])$ and let first $D_1$ and $D_2$ contain $R$. 
Let $E \in \binom{D_1\endo \cup D_2\endo}{k}$ and $E_i = E \cap D_i\endo$. Since $(\tau \circ Q')(E \cup D_1\exo \cup D_2\exo)$ is the union of the sets $(\tau \circ Q')(E_1 \cup D_1\exo)$ and $(\tau \circ Q')(E_2 \cup D_2\exo)$, its maximal value is equal to $a$ if and only if the maximal values of  
$(\tau \circ Q')(E_i \cup D_i\exo)$ for $i=1,2$ are less than or equal to $a$ and equality holds for at least one $i$. Hence, $\combineU(\P[Q',D_1], \P[Q',D_2])(a,k)$ can be computed by summing up $\P[Q',D_1](a_1,k_1) \cdot \P[Q',D_2](a_2,k_2)$ over all combinations with $k_1 + k_2 = k$ and $\max(a_1,a_2) = a$.

If $D_1$ and $D_2$ do not contain $R$, we can argue about the complement $\binom{\size{D'}}{k} - \P[Q',D']$ as follows. $Q'(E \cup D_1\exo \cup D_2\exo)$ is empty if and only if both $Q'(E_1 \cup D_1\exo)$ and $Q'(E_2 \cup D_2\exo)$ are empty. Hence, we sum up the product of the complements of $\P[Q',D_1](k_1)$ and $\P[Q',D_2](k_2)$ over all $k_1 + k_2 = k$.

The algorithm for $\combineX(\P[Q_1, D_1], \P[Q_2, D_2])$ is straightforward: 
For a Cartesian product to be nonempty, we need both sides to be nonempty and if $D_1$ contains $R$ (and $D_2$ does not), we need to attain $a$ as the maximum of $(\tau \circ Q_1)(E_1 \cup D_1\exo)$ and have $Q_2(E_2 \cup D_2\exo)$ nonempty.

Together, these algorithms show how to compute Shapley values for $\aggmax$-aggregation over \allhierarchical CQs in polynomial time and, hence, completes the proof of \Cref{thm:max-count-distinct-classification}.

\begin{remark}
A polynomial-time algorithm for \allhierarchical queries with $\aggmin$ and $\aggmax$ aggregation follows from the work of \citet{DBLP:journals/corr/abs-2506-16923}, who established, independently of our work, polynomial-time algorithms for $\aggmin$ and $\aggmax$ aggregation over d-trees. Such d-trees can, in turn, be obtained in polynomial time from a database and an \allhierarchical CQ, as shown in the context of knowledge compilation for query evaluation over probabilistic databases~\cite{DBLP:journals/pvldb/FinkHO12,DBLP:conf/sum/OlteanuH08}.
\end{remark}

\section{Average and Quantile}\label{sec:q}
In the previous section, we established that the tractability frontier for the aggregate functions min, max, and count-distinct coincides with the class of \allhierarchical CQs. In this section, we turn our attention to two other aggregate functions---average and quantile---and demonstrate that their Shapley value is inherently harder, as their tractability frontier is strictly narrower:
\begin{theorem}\label{thm:avg-quantile-classification}
Let $Q$ be a CQ without self-joins, $q\in(0,1)_{\mathbb{Q}}$ fixed, and $\alpha\in\set{\aggavg,\aggquantile_q}$. Denote by $A_\tau$ the aggregate query 
$\alpha\circ\tau\circ Q$.
\begin{enumerate}
    \item If $Q$ is \qhierarchical, then $\Shapley(f,A_\tau)$ is computable in polynomial time, given a database $D$ and a fact $f$, for every localized $\tau$.
    \item Otherwise, there is a localized $\tau$ such that computing $\Shapley(f, A_\tau)$, given a database $D$ and a fact $f$, is \fpsharpp-complete. 
\end{enumerate}
\end{theorem}

The notion of \qhierarchical was introduced by Berkholz, Keppeler, and Schweikardt~\cite{DBLP:conf/pods/BerkholzKS17}, and has subsequently been used to characterize the fine-grained complexity of query-evaluation tasks in several studies (e.g.,~\cite{DBLP:conf/pods/BerkholzKS17,DBLP:conf/icdt/MunozRV24,DBLP:conf/sigmod/IdrisUV17,DBLP:conf/pods/BerkholzM21,DBLP:journals/pvldb/WangHDY23}). To the best of our knowledge, \Cref{thm:avg-quantile-classification} is the first to show that the \qhierarchical property also delineates a boundary of coarse-grained complexity.
In the remainder of this section, we prove \Cref{thm:avg-quantile-classification}. We begin with the positive side.

\subsection{An Algorithm for Q-Hierarchical Queries}

We first give a closed formula for Shapley values for $\aggavg$-aggregation over the simplest case of a single relational query. It is tempting to think that the Shapley value of a fact $R(x)$ is $\frac{\tau(\vec x)}{n}$, but this is not the case, since the denominator of the average changes with the size of the subset considered. 

\begin{restatable}{proposition}
{propclosedformulaavg}
Let $Q(\vec x) \la R(\vec x)$ and assume that all facts in $R$ are endogenous. Let $n = \size{D}$ and $H(n) = \sum_{k = 1}^n \frac{1}{k}$ denote the $n$-th partial sum of the harmonic series. Then
    \[\Shapley(R(\vec t), \aggavg \circ \tau \circ Q) 
    = 
    \frac{H(n)}{n}\cdot \tau(\vec t) + \frac{H(n) - 1}{n(n - 1)} \sum_{R(\vec t') \neq R(\vec t)} \tau(\vec t')\text{\,.}\]
\end{restatable}

We omit the closed formula for Shapley values for $\aggquantile$ for $Q(x) \la R(x)$, since  it is too complicated to provide an advantage over a general algorithm. Instead, we turn our attention to the general algorithm by instantiating the generic algorithm from \Cref{fig:generic-alg} for \qhierarchical CQs $Q$. 

To take full advantage of the structure of \qhierarchical CQs, we need to specify the following: In step \ref{alg:choice-of-root}, if $Q'$ has several root variables, the algorithm chooses a \emph{free} root variable if there is one. If there are only existentially quantified root variables, then $Q'$ needs to be Boolean by the \qhierarchical property.

We next present a data structure $\P[Q',D']$ that allows us to compute $\sumword_k(\alpha \circ \tau \circ Q,D)$ for $\alpha = \aggavg$ as well as for $\alpha = \aggquantile_q$. It is defined as follows.
\begin{itemize}
    \item If $D'$ contains the relation $R$ that determines $\tau$, then $\P[Q',D']$ is a map that maps every quintuple $(a,k, \ell_{<}, \ell_{=}, \ell_{>})$ with $a \in (\tau \circ Q')(D')$, $0 \leq k \leq \size{D\endo}$, and $0 \leq \ell_{<} + \ell_{=} + \ell_{>} \leq \size{Q'(D')}$ to the number of subsets $E \in \binom{{D'}\endo}{k}$ with 
    the property that $(\tau \circ Q')(E \cup {D'}\exo)$ contains the element $a$ exactly $\ell_{=}$ times, exactly $\ell_{<}$ elements $b$ with $b < a$, and  exactly $\ell_{>}$ elements $c$ with $c > a$. (Recall that $\tau$ is defined on $Q'$, as explained at the beginning of \Cref{sec:template}.)
    \item Otherwise, $\P[Q',D']$ maps each pair $(k,\ell)$ with $0 \leq k \leq  \size{{D'}\endo}$ and $0 \leq \ell \leq \size{Q'(D')}$ to the number of subsets $E \in \binom{{D'}\endo}{k}$ such that $\size{Q'(E \cup {D'}\exo)} = \ell$.
\end{itemize}

To compute $\sumword_k(\aggavg \circ \tau \circ Q, D)$, we write the sum over a bag $B$ as the sum over all elements of the set underlying $B$ times their multiplicity in $B$. Using this in the formula for the average yields
\[
\sumword_k(\aggavg \circ \tau \circ Q, D) = \sum_{a \in (\tau \circ Q)(D)} \sum_{0 \leq \ell_< + \ell_= + \ell_> \leq \size{Q(D)}} \frac{a \cdot \ell_=}{\ell_< + \ell_= + \ell_> } \P[Q,D](a,k,\ell_<,\ell_=,\ell_>)\,.
\]

To compute $\sumword_k(\aggquantile_q \circ \tau \circ Q, D)$, we make the following trivial observation. $a$ is the $i$-th smallest element of a bag $B$ if $B$ contains less than $i$ elements smaller than $a$ and at least $i$ elements smaller than or equal to $a$. We further decompose $\aggquantile_q(B)$ into a sum over all elements of the set underlying $B$ times their contribution (either 0, 1 or $\tfrac{1}{2}$) to $\aggquantile_q(B)$.
Applying this to $B = (\tau \circ Q)(E \cup D\exo)$ 
allows us to compute $\sumword_k(\aggquantile_q \circ \tau \circ Q, D)$ via
\[
\sumword_k(\aggquantile_q \circ \tau \circ Q, D) = \sum_{a \in (\tau \circ Q)(D)} \;\sum_{0 \leq \ell_< + \ell_= + \ell_> \leq \size{Q(D)}} a \cdot f_q(\ell_<, \ell_=, \ell_>) \cdot \P[Q,D](a,k,\ell_<,\ell_=,\ell_>)\,.
\]
where we denote the indicator function of a condition $C$ by $\chi(C)$ to define $f_q(\ell_<, \ell_=, \ell_>)$ as
\begin{align*}
f_q(\ell_<, \ell_=, \ell_>) \defeq 
\frac{1}{2}\Big(
&\chi\big(\ell_< < \lceil q (\ell_< + \ell_= + \ell_>)\rceil\big) 
\cdot \chi\big(\ell_< + \ell_= \geq \lceil q (\ell_< + \ell_= + \ell_>)\rceil\big) \\
&+ \chi\big(\ell_< < \lfloor q (\ell_< + \ell_= + \ell_>) + 1\rfloor\big) \cdot \chi\big(\ell_< + \ell_= \geq \lfloor q (\ell_< + \ell_= + \ell_>) + 1\rfloor\big)
\Big) \,.
\end{align*}

We show how to compute $\combineU$ and $\combineX$ in polynomial time in \Cref{apx:q}. The algorithms are similar to those in the previous section with the addition that they use the fact that for non-Boolean $Q'$, the sets $Q'(D_1)$ and $Q'(D_2)$ are disjoint.

\subsection{Hardness}\label{sec:q-hardness}
We first show a reduction from the case of the following CQ, which is the simplest CQ that is \allhierarchical but not q-hierarchical:
\begin{equation}
\Qxyy(x) \la R(x,y), S(y)
\label{eq:qxyy}
\end{equation}
The following lemma gives a general reduction from $\Qxyy$. 

\begin{restatable}{lemma}
{lemmageneralreductionnonqhierarchical}
\label{lemma:general-reduction-non-q-hierarchical}
Consider an \AggCQ of the form $\alpha\circ\tau\circ\Qxyy$, and
let $Q_0$ be a CQ without self-joins, such that $Q_0$ is \allhierarchical but not q-hierarchical. 
There is a value function $\tau_0$, having the form $\tau \circ \tauidi{i}$, and a polynomial-time algorithm that, given an input database $D$ for $\Qxyy$, constructs an input database $D_0$ for $Q_0$ and a function $h:D\endo\rightarrow D_0\endo$ such that 
$$\Shapley(f, \alpha\circ\tau\circ\Qxyy)=
\Shapley(h(f), \alpha\circ\tau_0\circ Q_0)
$$
for every endogenous fact $f\in D\endo$. 
\end{restatable}

Hence, it suffices to prove the following, and we do so in \Cref{apx:q}. 
(Recall that the value functions $\taugb{b}$  and $\taurelu$ are defined in \Cref{eq:taurelu,eq:tauidi}, respectively.)

\begin{restatable}{lemma}
{lemmaavgquantilesimplesthard}
\label{lemma:avg-quantile-simplest-hard}
For $A=\aggavg\circ\taurelu\circ \Qxyy$ and
$A=\aggquantile_q\circ\taugb{0}\circ \Qxyy$ for fixed $q \in (0,1)_{\mathbb Q}$, computing $\Shapley(f, A)$, given a database $D$ and a fact $f$, is \fpsharpp-complete. 
\end{restatable}
\begin{proofsketch}
For $\aggavg$, the proof is by a reduction from the \e{\#Set-Cover} problem: given a set $\X$ of elements, a collection $\Y$ of subsets of $\X$, determine the number of collections $\C\subseteq\Y$ that cover all of $\X$ (i.e., $\X=\cup\C$). This problem is known to be \#P-complete~\cite{DBLP:journals/siamcomp/ProvanB83}. We construct an equation system over the variables $Z_{i,j}$ that, when solved, hold number of subsets of $\Y$ with $j$ members, covering precisely $i$ elements from $\X$. To construct this equation system, we define several databases $D'$ and use their $\Shapley(f,A)$ to construct one equation (similarly to past proofs of hardness of the Shapley value, e.g.,~\cite{DBLP:conf/stacs/AzizK14}). To crux of the proof is to show that the resulting matrix is invertible; we show it by establishing that the matrix is the Kronecker product of the Hilbert matrix ($M_{i,j}=\frac{1}{i+j-1}$) and the Hankel matrix with factorial entries ($M_{i,j}=(i+j)!$). 

For $\aggquantile$, we show how to construct a database $D$ so that $A(D)$ simulates (in an intuitive sense) the Set-Cover game (where the goal of a coalition is to cover all elements) for  $A=\aggquantile_q\circ\taugb{0}\circ \Qxyy$. Hardness for this game was established in \cite{gilad2024importanceparametersdatabasequeries}.
\end{proofsketch}

\section{Q-Hierarchy might not be Enough}\label{sec:duplicates}
Given the results of the previous section, one might wonder whether the restriction of being \qhierarchical suffices for the tractability of every (natural) aggregate function. More specifically, suppose that $\alpha$ is an aggregate function such that $\Shapley(f,A)$ is computable in polynomial time whenever the underlying CQ of $A$ has a single relation; does it imply the tractability of $\Shapley(f,A)$ when the underlying CQ is \qhierarchical without self-joins? 

This is not the case. In this section, we establish the negative answer by studying the aggregate function $\aggduplicates$ (has-duplicates) and showing that the tractability frontier of the Shapley value is a class stricter than that of the \qhierarchical CQs.

The stricter class of CQs is that of the \e{strongly \qhierarchical} CQs, or \e{\sqhierarchical} for short.
Let $Q$ be a CQ. We say that $Q$ is \sqhierarchical if $Q$ is \allhierarchical and $Q$ has no head variable whose set of atoms is strictly contained in that of another variable; that is, for all variables $x$ and $y$ of $Q$, if $\atoms(Q,x)\subsetneq\atoms(Q,y)$, then $x$ is existential. Note that an \sqhierarchical CQ is necessarily \qhierarchical. For illustration, the following CQs are \sqhierarchical:
$$
Q_1(x) \la R(x,y),S(x) 
\quad\vrule\quad
Q_2(x,y) \la R(x,y),S(x,y,z)
\quad\vrule\quad
Q_3(x,z) \la R(x,y),S(x),T(z)
$$
In contrast, the \qhierarchical CQ $Q_4(x,y) \la R(x,y),S(x)$ is \e{not} \sqhierarchical, since $\atoms(y)\subsetneq \atoms(x)$ and, yet, $y$ is a head variable.
In this section, we show that:

\begin{restatable}{theorem}
{thmduplicatesclassification}
\label{thm:duplicates-classification}
Let $Q$ be a CQ without self-joins.
\begin{enumerate}
    \item If $Q$ is \sqhierarchical, then $\Shapley(f,A)$ is computable in polynomial time for $A=\aggduplicates\circ\tau\circ Q$ on every localized $\tau$.
    \item Otherwise, there is a localized $\tau$ such that computing $\Shapley(f, A)$, given a database $D$ and a fact $f$, is \fpsharpp-complete. 
\end{enumerate}
\end{restatable}
\begin{proofsketch}
Tractability is shown by an adaptation of the generic algorithm (\Cref{fig:generic-alg}). We prove hardness as follows.
    \Cref{thm:livshits-hardness} covers all non-\exhierarchical CQs. We extend \Cref{thm:hardness-const-per-single} to the $\aggduplicates$ function, hence cover all CQs that are \exhierarchical but not \allhierarchical.
We then prove hardness for the following two \AggCQs, through a reduction from the problem of calculating the permanent of a 0/1-matrix, which is known to be \#P-hard due to Valiant's Theorem~\cite{DBLP:journals/tcs/Valiant79}:
\begin{enumerate}
    \item $\aggduplicates\circ\taurelu\circ \Qxyy$ \;for\; $\Qxyy(x) \la R(x,y), S(y)$.
    \item $\aggduplicates\circ\tauidi{1}\circ\QFxyy$ \;for\; $\QFxyy(x,y) \la R(x,y), S(y)$.
\end{enumerate}

From hardness (1) and \Cref{lemma:general-reduction-non-q-hierarchical}, we conclude hardness for every CQ that is \allhierarchical but not \qhierarchical. Next, we prove a general reduction similar to \Cref{lemma:general-reduction-non-q-hierarchical}, stating that if
$Q$ is \qhierarchical but not \sqhierarchical, then the computation of $\Shapley(f, \aggduplicates\circ\taurelu\circ \QFxyy)$ reduces to 
$\Shapley(f, \aggduplicates\circ\tau\circ Q)$ for some localized value function $\tau$; hence, hardness (2) and the new general reduction jointly yield hardness for every CQ that is \qhierarchical but not \sqhierarchical, completing the story.
The details are in \Cref{apx:q}.
\end{proofsketch}

\section{Discussion: Towards Finer Classifications}\label{sec:finer}

The dichotomy theorems established in the previous sections (namely, \Cref{thm:max-count-distinct-classification,thm:avg-quantile-classification,thm:duplicates-classification}) classify the \AggCQs for each aggregation function $\alpha$ into two categories: the \e{tractable} ones, which are tractable for \e{every} localized value function $\tau$, and the \e{intractable} ones, which are \fpsharpp-complete for \e{at least one} value function $\tau$, which is localized on some predefined atom $R(\vec z)$ (or a constant, as a special case). This classification, however, does not rule out the possibility that an \AggCQ that is intractable under one value function may become tractable under a different value function $\tau$, and even a natural one such as $\tauidi{i}$. Thus, there is potential for finer classifications that take $\tau$ into account. In this section, we explore this challenge and offer preliminary insights.

\subsection{Choice of a Value Function}\label{sec:finer-function}

In the case of $\aggsum$ and $\aggcount$ (\Cref{thm:livshits-hardness}), the classification applies to every function $\tau$, as long as the resulting \AggCQ is not constant~\cite{DBLP:journals/lmcs/LivshitsK22}. For other aggregation functions, the choice of $\tau$ can definitely make a difference. For illustration, we inspect what happens when we replace the value function $\tau$ with the natural copying function $\tauid$. 

We have seen in the proof sketch of \Cref{thm:duplicates-classification} the hardness of
$\aggduplicates\circ\taurelu\circ \Qxyy$. Replacing $\taurelu$ with $\tauid$ would make the query constantly zero, and the Shapley value trivial, for every relational query $Q$, since $Q(D)$ cannot include duplicates by definition. 
As another example, consider the \exhierarchical CQ $Q(x) \la R(x), S(x,y),T(y)$ and the \AggCQ $\aggcdis \circ \tau\circ Q$ for a constant value function. Computing the Shapley value is intractable, by \Cref{thm:hardness-const-per-single}. However, changing  $\tau$ to $\tauid$ makes it solvable in polynomial time: since $\tauid$ is injective, this \AggCQ is equivalent to $\aggcount \circ \taugb{0} \circ Q$, where the Shapley value is tractable (since $\aggcount$ is tractable for every \exhierarchical CQ without self-joins~\cite{DBLP:journals/lmcs/LivshitsBKS21}).

Contrasting the above examples, we next show that our hardness results for the other aggregation functions are robust to replacing the value function with $\tauid$.
More generally, we show that every $\tau$ can be replaced by $\tauidi{i}$, as long as $\tau$ is a monotonically increasing (i. e. non-decreasing) function over $\tauidi{i}$, which applies $\taugbi{b}{i}$, $\taurelui{i}$ and also $\tau \equiv c$.

\begin{restatable}{theorem}
{thmhardnessmonotonic}
\label{thm:hardness-monotonic}
 Let $\alpha\circ\tau\circ Q$ be an \AggCQ where
$\alpha\in\set{\aggmin,\aggmax,\aggavg, \aggquantile_q}$ and 
$Q$ has no self-joins. Suppose that  $\tau = \gammamon \circ \tauidi{i}$ for some number $i$ and a monotonically increasing function $\gammamon \colon \reals \to \reals$. 
If $\Shapley(f, \alpha\circ\tau\circ Q)$ is \fpsharpp-complete to compute, then so is $\Shapley(f, \alpha\circ\tauidi{i}\circ Q)$.
\end{restatable}

Hence, 
From \Cref{thm:hardness-monotonic} and the hardness results of the previous sections, 
we conclude that:
\begin{corollary}
\label{cor:robust-to-tauid}
Let $A=\alpha\circ\tauidi{i}\circ Q$ be an \AggCQ where $\alpha\in\set{\aggmin,\aggmax,\aggavg,\aggquantile_q}$. If $Q$ has no self-joins and $Q$ is not \allhierarchical, then
computing 
$\Shapley(f, A)$ is \fpsharpp-complete.
\end{corollary}

\subsection{Localization of a Value Function}\label{sec:finer-localization}

Next, we ask whether the choice of the atom $R(\vec z)$ on which $\tau$ is localized matters; that is, can it be the case that the hard \AggCQ becomes tractable if $\tau$ is assumed to be localized on a specific $R$? 

For the aggregation functions $\aggsum$ and $\aggcount$, the hardness of \Cref{thm:livshits-hardness} is robust to the localization (as long as $\alpha\circ\tau$ is not constant). For $\aggmin$, $\aggmax$ and $\aggcdis$, hardness is covered by \Cref{thm:hardness-const-per-single}, which is also robust to the localization, in a trivial sense (as the constant value function is localized on every atom). It continues to hold (in a non-trivial sense) for the generalization of \Cref{cor:robust-to-tauid}. Next, we show that our hardness results for $\aggavg$, $\aggquantile_q$, and $\aggduplicates$ do not possess this robustness---they may become false if we replace the localization atom $R(\vec z)$.

Consider the following \AggCQs:
\begin{enumerate}
    \item $\aggavg\circ\taurelui{1}\circ\Qxyyz$ \;for\;  $\Qxyyz(x,z) \la R(x,y), S(y), T(z)$.
    \item $\aggmed\circ\taugbi{0}{1}\circ\Qxyyz$.
    \item $\aggduplicates\circ\tauidi{1}\circ\QFxyy$ \;for\; $\QFxyy(x,y) \la R(x,y), S(y)$. 
\end{enumerate}
In each of the three, the computation of the Shapley value is \fpsharpp-complete. 
For the first two, it can be shown by immediate reductions from $\aggavg\circ\taurelu\circ\Qxyy$ and 
    $\aggmed\circ\taugb{0}\circ\Qxyy$, respectively, (eliminating the effect of $T(z)$ by assigning to $T$ a relation with a single, exogenous tuple) that were established to be hard in \Cref{lemma:avg-quantile-simplest-hard}.
The third \AggCQ was shown to be hard in the proof sketch of \Cref{thm:duplicates-classification}. In contrast, for each of the three, the Shapley value is computable in polynomial time if we change the localization assumption from the first atom to the last atom:

\begin{restatable}{proposition}
{proplocalizationsinptime}
\label{prop:localizations-in-ptime}
For each of the following \AggCQs, there is a polynomial-time algorithm for computing $\Shapley(f, A)$, given a database $D$ and a fact $f$.
\begin{enumerate}
\item $\aggavg\circ\taurelui{2}\circ\Qxyyz$ \;for\;  $\Qxyyz(x,z) \la R(x,y), S(y), T(z)$.
\item $\aggmed\circ\taugbi{0}{2}\circ\Qxyyz$.
    \item $\aggduplicates\circ\tauidi{2}\circ\QFxyy$ \;for\; $\QFxyy(x,y) \la R(x,y), S(y)$.
\end{enumerate}
\end{restatable}

Hence, a valuable challenge for future research is to establish finer classifications that account for the identity of the atom on which localization is assumed.

\subsection{Beyond Localized Value Functions}\label{sec:finer-beyond-localized}

The polynomial-time algorithms presented so far crucially rely on the locality of the value function~$\tau$. In contrast, the case of non-localized value functions remains largely unexplored and is left as an open direction for future research. In this section, we outline the main challenges that arise when attempting to extend our results to this more general setting.

The algorithm of Livshits et al.~\cite{DBLP:journals/lmcs/LivshitsBKS21} for $\aggsum$ over \exhierarchical CQs applies to general (not necessarily localized) value functions $\tau$, as long as $\tau$ can be computed in polynomial time. The reason for that is simple---the algorithm considers one answer at a time, independently of the others (due to the linearity of expectation), so it can simply compute $\tau$ upfront based on the considered answer, regardless of the number of atoms that determine $\tau$. 

However, when going beyond $\aggsum$, things get more complicated. For illustration, let us consider the case of $\aggmin$ and $\aggmax$. The algorithm of \Cref{sec:all} for $\aggmin$/$\aggmax$ over \allhierarchical CQs extends to a non-localized $\tau$, as long as it is a \e{monotonic monoid}. More precisely, if we assume that $\tau$ is the result 
$x_1\otimes\dots\otimes x_\ell$ of applying a non-decreasing monoid $\otimes:\reals^2\rightarrow\reals$ on a sequence $(x_1,\dots,x_\ell)$ of numeric free variables, then we can extend the algorithm for $\Shapley(f,A)$ to handle $\aggmax\circ\tau$.
For example, this applies to $\aggmax(x_1+x_2)$, $\aggmax(\max(x_1,x_2))$, and so on. (Similarly, we can support $\aggmin$ with a non-increasing monoid.) The extension is as follows. For each subquery $Q'$ and relevant number $a$, we store in $\P$ the number of subsets $E$ of $D\endo$ where the monoid, restricted to the variables of $Q'$, has $a$ as the maximal value. Due to the monotonicity, we can then realize the $\combineX$ function efficiently by considering all combinations of numbers $a'$ and $a''$ with $a'\otimes a''=a$. (We will give the complete argument in the extended version of this manuscript.)

Importantly, it is necessary to restrict $\tau$ (e.g., being a monotonic monoid) to extend our upper bound for min and max; it is insufficient to assume only that $\tau$ is computable in polynomial time (as in the case of sum). To see why, consider the Boolean CQ
$Q_0() \la R_0(x), S_0(x,y), T_0(y)$.
Suppose that all facts of $S_0$ are exogenous. In this case, computing $\Shapley(Q_0,f)$ is \fpsharpp-hard~\cite{DBLP:journals/lmcs/LivshitsBKS21}. Now, %
consider the \AggCQ 
\[A\defeq \aggmax\circ \tau\circ Q \mbox{\, for\, $Q(x,y) \la R(x), T(y)$}
\,.\]

Define $\tau$ as follows. Assume that $\tau(x,y)$ is equal to $1$ if $x$ divides $y$ and $0$ otherwise. Then $A$ can \e{simulate} the Boolean CQ $Q_0$ for the computation of the Shapley value with the following database $D$ that we construct for $A$. Given a database $D_0$ for $Q_0$, choose a distinct prime $p_c$ for each $c \in D_0^{R_0}$ and add it to $D^R$. In $D^T$, we store for each $d \in D_0^{T_0}$ the product over all $p_c$ with $(c,d) \in D_0^{S_0}$ (or a power of another prime $p'$ if no such $c$ exists). This construction can be performed in polynomial time, since the $n$th prime is in $O(n \log{}n)$ \cite{Hua1982Introduction}. Then, for each subset $E_0$ of $D_0\endo$, it holds that $Q(D_0\exo\cup E_0)$ is true if and only if $A(D\exo\cup E)=1$, where $E$ is the set of facts of $D$ that corresponds to $E_0$. Therefore, each fact $f$ of $D_0$ has the same Shapley value as its corresponding fact of $D$ .

We conclude that there are polynomial-time computable functions $\tau$ for which the Shapley value for the min/max is intractable even for the Cartesian-product CQ (which is \allhierarchical and even \sqhierarchical). We conjecture that a similar situation applies to every aggregate function that we considered in this manuscript, with the exception of sum and count. Consequently, future extensions of this work for non-local $\tau$ will need to make some assumptions on $\tau$.

\section{Concluding Remarks}\label{sec:conclusions}

We investigated the complexity of computing the Shapley value of a tuple for \AggCQs $\alpha\circ\tau\circ Q$ where the underlying CQ $Q$ has no self-joins and the value function $\tau$ is localized. 
While the cases of sum and count have been resolved in prior work~\cite{DBLP:journals/lmcs/LivshitsBKS21}, our results shed light on the tractability frontiers of
other common aggregate functions, including count-distinct, min, max, average, quantile, and has-duplicates.
In particular, we showed that these functions capture different classes of hierarchical CQs, as shown in \Cref{fig:hierarchy}.

Many avenues remain for future continuation of this research. Following the discussion on \Cref{sec:finer}, an important direction is to pursue classification with respect to general classes of value functions $\tau$. 
Another important direction is to extend the class of relational queries in the considered aggregate queries, namely the class of CQs without self-joins. While handling CQs with self-joins is traditionally a daunting challenge in the case of the Shapley value (and other complexity classifications, e.g.,~\cite{DBLP:journals/pacmmod/PadmanabhaSS24,DBLP:conf/pods/CarmeliS23,DBLP:journals/jacm/DalviS12}), we can also explore the extension to the classes of the \e{constant-free unions of connected CQs}, and \e{connected homomorphism-closed graph queries}, where dichotomies for the Boolean Shapley value have been established by  Bienvenu, Figueira, and Lafourcade~\cite{DBLP:journals/pacmmod/BienvenuFL24}.

In addition to the directions of extending $Q$ and $\tau$, we are interested in generalizing our study from specific (though commonly used) aggregate functions $\alpha$ to broad classes of aggregate functions that are identified by general properties (as exemplified in \Cref{thm:hardness-const-per-single}) that cover the aggregate functions of standard SQL and possibly many others.
Specifically, an important future direction is to extend our results into full classifications that cover aggregate functions from wide classes of functions defined by intuitive properties.

Mirroring past applications of the Shapley value to databases~\cite{DBLP:journals/lmcs/LivshitsK22, DBLP:journals/lmcs/LivshitsBKS21,DBLP:conf/icdt/GroheK0S24}, an important direction is to study the complexity of \e{approximating} the Shapley value with error guarantees. This is highly relevant to the pragmatics of calculating the Shapley value in reality, as well as the usage of SAT-solvers and circuits that were shown to provide practical solutions for 
both ordinary CQs~\cite{DBLP:conf/sigmod/DeutchFKM22,DBLP:journals/sigmod/AmarilliC24} and \AggCQs for certain aggregation functions~\cite{DBLP:journals/corr/abs-2506-16923}.

There are also directions that concern different database models. The work of \citet{DBLP:journals/pacmmod/KarmakarMSB24} gives rise to a variant of the problem involving \e{probabilistic databases}: they studied the expected contribution score (for the general class of Shapley-like scores) of an endogenous fact when tuples are probabilistic (in the \e{tuple-independent model}~\cite{DBLP:conf/vldb/DalviS04}), and this problem generalizes naturally to the expected contribution to an aggregate query.
In this vein, another direction regards the computation of the Shapley value for queries over ontologies (known as ontology-mediated query answering), 
where ontological axioms together with the database induce a set of possible worlds,
as done recently by \citet{DBLP:conf/kr/BienvenuFL24}. In their work, the utility function for a Boolean query is the certainty of the query. For aggregate queries, we can adopt known semantics from the literature, such as those proposed by~\citet{DBLP:conf/cikm/CalvaneseKNT08} or the range semantics by \citet{DBLP:conf/pods/AfratiK08}.

Finally, it is also important to explore the rationality of using the Shapley value for aggregate queries over databases, in continuation of the discussion in \Cref{remark:rationality}. Different application scenarios may call for rationality criteria other than those of the Shapley value, and it is therefore important for a database system to support a broad collection of attribution measures. 
On this matter, 
\citet{DBLP:journals/pacmmod/BienvenuFL25} argued that, in the case of Boolean conjunctive queries (or unions thereof), attribution should reflect rationality criteria based on the set of homomorphisms that include a given fact. They proposed the weighted sums of minimal supports (WSMS) measure, which satisfies these criteria.
We leave as an important avenue for future work the study of rationality criteria for \AggCQs (and other types of numerical aggregate queries); this includes understanding whether and how the axiomatic properties advocated by Bienvenu et al.~should extend or be adapted to such queries, and developing corresponding attribution measures and algorithms.

\begin{acks}
We are grateful to the PODS referees for providing insightful feedback and contributing excellent suggestions to this paper.
This work was supported by the German Research Foundation (DFG) grants GR 1492/16-1 and KI 2348/1-1 (DIP Program). Christoph Standke was funded by the European Union (ERC, SymSim, 101054974) and by the German Research Foundation (GRK 2236, UnRAVeL). Views and opinions expressed are however those of the author(s) only and do not necessarily reflect those of the European Union or the European Research Council. Neither the European Union nor the granting authority can be held responsible for them. 
\end{acks}

\bibliographystyle{ACM-Reference-Format}
\bibliography{main}

\appendix
\crefalias{section}{appendix}

\section{Properties of the Shapley value}\label{apx:shapley}

In this section, we review some basic properties of the Shapley value which we will later use in our proofs. Let $(P,\nu)$ be a cooperative game (with $\nu(\emptyset) = 0$). Then the Shapley 
value has the following properties \cite{Shapley}:\footnote{\label{footnote:axioms}The properties we list here are not the original axioms proposed by \citet{Shapley}. In particular, the original axiom of symmetry is stronger. However, these properties together also characterize the Shapley value uniquely and are therefore often treated as axioms for the Shapley value (see, e.g., \cite{maschler2020game,9565902,molnar2020interpretable}).}
\begin{itemize}
    \item Efficiency: The sum of the Shapley values of all player is the utility of the whole set of players. In formula, $\sum_{p \in P} \Shapley(p, \nu) = \nu(P)$.
    \item Linearity: The Shapley value is linear in the utility function; that is, $\Shapley(p, \nu_1 + c \cdot \nu_2) = \Shapley(p, \nu_1) + c \cdot \Shapley(p, \nu_2)$.
    \item Symmetry: If two players $p_1$ and $p_2$ are interchangeable w.r.t $\nu$, then their Shapley values are equal. That is, if $\nu(C \cup \set{p_1}) = \nu(C \cup \set{p_2})$ for each $C \subseteq P \setminus \set{p_1, p_2}$, then $\Shapley(p_1, \nu) = \Shapley(p_2, \nu)$.
    \item Null player property: If a player does not contribute to any coalition, its Shapley value is zero. That means, that if $\nu(C \cup \set{p}) = \nu(C)$ for each $C \subseteq P \setminus \set{p}$, then $\Shapley(p, \nu) = 0$. Furthermore, this player can be removed from the game without changing the Shapley values of all other players.
\end{itemize}

\section{Omitted Proofs from \Cref{sec:machinery}}
In this section, we present the proofs omitted from \Cref{sec:machinery}.

\propconstantsameasboolean*
\begin{proof}
    This claim follows from the linearity property of the Shapley value. To see that we can apply this property, let $C \subseteq D\endo$. Since $\alpha$ is constant per singleton, $A$ returns $\alpha(\multiset{c})$ if $Q(C \cup D\exo)$ is nonempty and $0$ otherwise. In the same way, $\Qbool$ returns $1$ if $Q(C \cup D\exo)$ is nonempty and $0$ otherwise, so we have $A(C \cup D\exo) = \alpha(\multiset{c}) \cdot \Qbool(C \cup D\exo)$.
\end{proof}

\section{Omitted Details from \Cref{sec:all}}\label{apx:all}

In this section, we give the missing proofs and details from \Cref{sec:all}.

\subsection{Count Distinct}

We start by proving the closed formula from the following proposition. Recall that $\chi_a$ was defined as $\chi_a(B) = 1$ if $a \in B$ and $0$ otherwise.
\propclosedformulacountdisctinct*
\begin{proof}
We can infer this formula directly from the properties of the Shapley value presented in \Cref{apx:shapley}. Before considering $\nu = A$, let us first consider $\nu_a \defeq \chi_a \circ \tau \circ Q$ for $a \in (\tau \circ Q)(D)$. Each fact $R(\vec x)$ with $\tau(\vec x) \neq a$ is a null player for this utility function, and all other players are interchangeable and therefore have equal Shapley values by symmetry. Since the sum of all Shapley values for $\nu_a$ is $1$ by efficiency, we obtain $\Shapley(R(\vec t), \nu_a)(D) = \frac{1}{\size{\set{R(\vec t') \, \mid \, \tau(\vec t') = \tau(\vec t)}}}$ if $\tau(\vec t) = a$ and $0$ otherwise. Writing $\nu = \sum_a \nu_a$ and using linearity yields the claim.
\end{proof}

The next lemma relates the indicator variables $\chi_a$ to the Boolean version $\Qbool$ of $Q$. 
\begin{lemma}\label{lemma:indicator-boolean}
    Let $Q$ be a CQ without self-joins and $\tau$ be localized on the atom $R(\vec z)$. 
    For a database $D$ and $a \in \reals$, let $D_a$ denote the database that is obtained from $D$ by removing all facts $R(\vec b)$ with $\tau(\vec b) \neq a$. Furthermore, let $\Qbool() \la Q(\vec x)$. Then
    \[
    \Qbool(D_a) = 1 \iff a \in (\tau \circ Q)(D).
    \]
\end{lemma}

\begin{proof}
    If $\Qbool(D_a) = 1$, then there is a homomorphism $h$ from $\Qbool$ (and hence from $Q$) to $D_a$. Since $h$ is a homomorphism, we have $R(h(\vec z)) \in D_a$ and hence $a \in \tau(h(\vec x))$. For the converse, assume $a \in (\tau \circ Q)(D)$. Then there is a homomorphism $h$ from $Q$ to $D$ with $h(\tau(\vec x)) = a$. Since $\tau$ is determined by $R(\vec z)$, we have $R(h(\vec z)) \in D_a$ and hence $h$ is a homomorphism from $\Qbool$ to $D_a$, so $\Qbool(D_a) = 1$.
\end{proof}
 With that established, we are now ready to proof the following lemma.
\lemreductioncdisboolean*
\begin{proof}
    Since the Shapley value is linear and $\aggcdis$ is the sum of the indicator variables $\chi_a$, we have 
    \[
    \Shapley(f,A)[D] = \sum_{a \in (\tau \circ Q)(D)} \Shapley(f,\chi_a \circ \tau \circ Q)[D].
    \]
    All facts in $D \setminus D_a$ are null players for $\chi_a \circ \tau \circ Q$, since they can only add values different from $a$ to $\tau \circ Q$. This shows that their Shapley values are $0$ and that removing them from the game does not change the Shapley values of the other players. Together with \Cref{lemma:indicator-boolean}, we find $\Shapley(f,\chi_a \circ \tau \circ Q)[D] = \Shapley(f,\chi_a \circ \tau \circ Q)[D_a] = \Shapley(f,Q')[D_a]$. This concludes the proof.
\end{proof}

\subsection{Min and Max}\label{sec:minmax}

 We first prove the closed formula given in the following proposition.

 \propclosedformulamax*

\begin{proof}
    When we add $f = R(\vec{t})$ to $C \subseteq D \setminus \set{f}$, only two things can happen. Either the maximum before adding $f$ was at least $\tau(\vec t)$ and the maximum does not change, or the new maximum will be $\tau(\vec t)$. The coalitions $C$ for which the first case holds contribute nothing to the Shapley value, while the contribution of a coalition for which the second case holds depends on the maximum of $(\tau \circ Q)(C)$. The idea is to decompose the formula for the Shapley value into all possible values for that maximum. Treating the case $k = 0$ (and hence $C = \emptyset$) seperately, we obain
    \begin{align*}
    &\Shapley(R(\vec t), \aggmax \circ \tau \circ Q)(D) \\ 
    &= 
    \frac{\tau(\vec t)}{\size{D}} + \sum_{k = 1}^{\size{D} - 1} q_k \sum_{\substack{a \in (\tau \circ Q)(D) \\ a < \tau(\vec t)}}  (\tau(\vec t) - a)  \cdot \size[\big]{\set{C \in D\setminus\set{f} \colon \size{C} = k \wedge A(C) = a}} \text{.} \\
    \end{align*}
    We now need to count the sets $C \subseteq D \setminus \set{f}$ of size $k$ with $(\aggmax \circ \tau \circ Q)(D) = a$. These are the sets that contain at least one of the $\mult{=}{a}$ facts $R(\vec t)$ with  $\tau(\vec t) = a$ and the $\tau$-value of every fact in $D$ must be $\leq a$. Summing over all possibilites to choose exactly $\ell$ facts $R(\vec t)$ with $\tau(\vec t) = a$ and $k - \ell$ with $\tau(\vec t) < a$, the number of these sets is equal to \[
    \sum_{\ell = 1}^{k} \binom{\mult{=}{a}}{\ell} \binom{\mult{<}{a}}{k-\ell} =
    \binom{\mult{=}{a} + \mult{<}{a}}{k} - \binom{\mult{=}{a}}{0}\binom{\mult{<}{a}}{k} 
    = \binom{\mult{\leq}{a}}{k} - \binom{\mult{<}{a}}{k}\text{,}\]
    using the combinatorial identity $\binom{a + b}{k} = \sum_{\ell = 0}^k \binom{a}{\ell}\binom{b}{k-\ell}$. Plugging the last equation into the one before and reordering the sums yields the claim.
\end{proof}
\subsubsection*{Details on the instantiation of the generic algorithm}

Next, we present the missing details about the algorithms for $\combineU$ and $\combineX$ for the data structure defined in \Cref{sec:minmax}. To simplify notation, we set $\P[Q_i,D_i](a,k) = 0$ whenever $a \notin (\tau \circ Q_i)(D_i)$ or $k \notin \set{0,1,\ldots,\size{D_i\endo}}$.

Let us first discuss the algorithm for  $\combineU(\P[Q', D_1], \P[Q', D_2])$ and let first $D_1$ and $D_2$ contain $R$. 
Let $E \in \binom{D_1\endo \cup D_2\endo}{k}$ and $E_i = E \cap D_i\endo$. Since $(\tau \circ Q')(E \cup D_1\exo \cup D_2\exo)$ is the union of the sets $(\tau \circ Q')(E_1 \cup D_1\exo)$ and $(\tau \circ Q')(E_2 \cup D_2\exo)$, its maximal value is equal to $a$ if and only if the maximal values of  
$(\tau \circ Q')(E_i \cup D_i\exo)$ for $i=1,2$ are less than or equal to $a$ and equality holds for at least one $i$. Hence, $\combineU(\P[Q',D_1], \P[Q',D_2])(a,k)$ can be computed by summing up $\P[Q',D_1](a_1,k_1) \cdot \P[Q',D_2](a_2,k_2)$ over all combinations with $k_1 + k_2 = k$ and $\max(a_1,a_2) = a$. This yields
    \begin{align*}
        \combineU(\P[Q',D_1], \P[Q',D_2])(a,k) = 
        &\sum_{k_1 = 0}^{k} \sum_{a_2 \leq x} \P[Q',D_1](a,k_1) \cdot   \P[Q',D_2](a_2, k - k_1) \\
        &+ \sum_{k_1 = 0}^{k} \sum_{a_1 < a}
         \P[Q',D_1](a_1, k_1)  \cdot \P[Q',D_2](a, k - k_1) \text.
    \end{align*}

If $D_1$ and $D_2$ do not contain $R$, we can argue about the complement $\binom{\size{D'}}{k} - \P[Q',D']$ as follows. $Q'(E \cup D_1\exo \cup D_2\exo)$ is empty if and only if both $Q'(E_1 \cup D_1\exo)$ and $Q'(E_2 \cup D_2\exo)$ are empty. Hence, we sum up the product of the complements of $\P[Q',D_1](k_1)$ and $\P[Q',D_2](k_2)$ over all $k_1 + k_2 = k$. Using the identity $\binom{n_1 + n_2}{k} = \sum_{k_1} \binom{n_1}{k_1}\binom{n_2}{k-k_1}$, this reduces to
    \begin{align*}
        \combineU&(\P[Q',D_1], \P[Q',D_2])(k) = 
        \sum_{k_1 = 0}^{k} \Big(\P[Q',D_1](k_1) \cdot \binom{\size{D_2\endo}}{k-k_1} \\ &+ \binom{\size{D_2\endo}}{k_1} \cdot \P[Q',D_2](k-k_1) - \P[Q',D_1](k_1) \cdot \P[Q',D_2](k-k_1)\Big)\text.
    \end{align*}

The algorithm for $\combineX(\P[D_1, Q_1], \P[D_2, Q_2])$ is straightforward: 
For a Cartesian product to be nonempty, we need both sides to be nonempty. So if neither $D_1$ nor $D_2$ contain $R$, we have
    \[
        \combineX(\P[Q_1,D_1], \P[Q_2,D_2])(k) = \sum_{k_1 = 0}^{k} \P[Q_1,D_1](k_1) \cdot \P[Q_2,D_2](k - k_1).
    \]
If $D_1$ contains $R$ (and $D_2$ does not), we need to attain $a$ as the maximum of $(\tau \circ Q_1)(E_1 \cup D_1\exo)$ and have $Q_2(E_2 \cup D_2\exo)$ nonempty, so
    \[
        \combineX(\P[Q_1,D_1], \P[Q_2,D_2])(a, k) = \sum_{k_1 = 0}^{k} \P[Q_1,D_1](a, k_1) \cdot \P[Q_2,D_2](k - k_1).
    \]

\section{Omitted Proofs from \Cref{sec:q}}\label{apx:q}
We now give the missing proofs and details from \Cref{sec:q}.

\subsection{Algorithms}
 We first proof the closed formula for Shapley values for $\aggavg$-aggregation.

\propclosedformulaavg*

\begin{proof}
    Since this proof is heavy on computations, let us first simplify the notation. Let $D = \set{f_1, \ldots, f_n}$ and $f_i = R(\vec t_i)$.  We set $y_i \defeq \tau(\vec t_i)$ and $A = \aggavg \circ \tau \circ Q$. Let $[n] = \set{1, \ldots, n}$. With this notation, we will access the facts we consider via their indices and therefore consider index sets $J$ instead of their corresponding coalitions $C$.
    
    The empty set is a special case for the average, which we need to treat separately. For that, we observe $\nu_{\avg}(\multiset{y_i}) - \nu_{\avg}(\emptyset) = y_i$ and $q_0 = \frac{1}{n}$. Now,  a short calculation yields
    \begin{align*}
        \Shapley(f_i, A) 
        &= \sum_{k=0}^{n-1} q_k \sum_{\substack{J \subseteq [n] \setminus \set{i} \\ \size{J} = k}}
        \big(
        \avg(J \cup \set{i}) - \avg(J) 
        \big) \\
        &= \frac{y_i}{n} + \sum_{k=1}^{n-1} q_k 
        \sum_{\substack{J \subseteq [n] \setminus \set{i} \\ \size{J} = k}}
        \bigg(
        \frac{1}{k + 1}\sum_{j \in J \cup \set{i}} y_j 
        - \frac{1}{k}\sum_{j \in J} y_j 
        \bigg) \\
        &= \frac{y_i}{n} + \sum_{k=1}^{n-1} q_k
        \sum_{\substack{J \subseteq [n] \setminus \set{i} \\ \size{J} = k}} 
        \bigg(
        \frac{y_i}{k + 1}  
        - \frac{1}{k(k+1)}\sum_{j \in J} y_j 
        \bigg) \\
        &= \frac{y_i}{n} + \sum_{k=1}^{n-1} q_k \cdot \Bigg(
        \bigg( \frac{y_i}{k + 1}
        \sum_{\substack{J \subseteq [n] \setminus \set{i} \\ \size{J} = k}} 1 \bigg) +
        \bigg( 
        \frac{1}{k(k+1)} \sum_{\substack{J \subseteq [n] \setminus \set{i} \\ \size{J} = k}}
        \sum_{j \in J} y_j 
        \bigg) \Bigg) \\
        &= \star
    \end{align*}

    Now, we observe that there are $\binom{n-1}{k}$ subsets $J$ of $[n] \setminus \set{i}$ of size $k$ and each element $j \in [n] \setminus \set{i}$ appears in exactly $\binom{n-2}{k-1}$ of them. Using this, we continue from the last line of the chain of equations to obtain

    \begin{align*}
        \star
        &= \frac{y_i}{n} + \sum_{k=1}^{n-1} q_k \cdot \Bigg(
        \bigg( \frac{y_i}{k + 1}
        \binom{n-1}{k} \bigg) +
        \bigg( 
        \frac{1}{k(k+1)} \binom{n-2}{k-1}
        \sum_{j \in [n] \setminus \set{i}} y_j 
        \bigg) \Bigg)
        \\
        &= \frac{y_i}{n} + y_i \cdot \sum_{k=1}^{n-1} \frac{k!(n-k-1)!}{n!}\frac{1}{k+1}\frac{(n-1)!}{k!(n-k-1)!} \\
        &{} \qquad - \bigg(\sum_{j \in [n] \setminus \set{i}} y_j\bigg) \sum_{k=1}^{n-1} \frac{k!(n-k-1)!}{n!}\frac{1}{k(k+1)}\frac{(n-2)!}{(k-1)!(n-k-1)!} \\
        &= \frac{y_i}{n} + y_i \cdot \sum_{k=1}^{n-1}\frac{1}{n(k+1)} + \bigg(\sum_{j \in [n] \setminus \set{i}} y_j\bigg) \cdot \sum_{k=1}^{n-1} \frac{1}{n(n-1)(k+1)}\\
        &= \frac{H(n)}{n} \cdot y_i - \frac{H(n) - 1}{n(n-1)}\sum_{j \in [n] \setminus \set{i}} y_j
    \end{align*}
This completes the proof.
\end{proof}

\subsubsection*{Details on the instantiation of the generic algorithm}

We give further details on the computation for $\sumword_k(A,D)$ for $\aggavg$-aggregation and $\aggquantile_q$-aggregation.

As explained in \Cref{sec:q}, To compute $\sumword_k(\aggavg \circ \tau \circ Q, D)$, we write the sum over a bag $B$ as the sum over all elements of the set underlying $B$ times their multiplicity in $B$. Using this in the formula for the average yields
\begin{align*}
\sumword_k(\aggavg \circ \tau \circ Q, D) &=  \sum_{E \in \binom{D\endo}{k}} (\aggavg \circ \tau \circ Q)(E \cup D\exo)\\
&= \sum_{E \in \binom{D\endo}{k}} \sum_{a \in (\tau \circ Q)(D)} \frac{a \cdot \mult{a}{=}((\tau \circ Q)(E \cup D\exo)))}{\size{(\tau \circ Q)(E \cup D\exo))}}\\
& = \sum_{a \in (\tau \circ Q)(D} \sum_{0 \leq \ell_< + \ell_= + \ell_> \leq \size{Q(D)}} \frac{a \cdot \ell_=}{\ell_< + \ell_= + \ell_> } \P[D,Q](a,k,\ell_<,\ell_=,\ell_>).    
\end{align*}

The formula for $\sumword_k(\aggquantile_q \circ \tau \circ Q, D)$ is obtained with a similar reasoning. Now, it remains to show how we can compute in polynomial time the functions $\combineU$ and $\combineX$ for $\P$ as defined in \Cref{sec:q}. 

For computing $\combineU(\P[Q',D_1],\P[Q',D_2])$, there are three different cases to consider:
\begin{enumerate}
    \item $Q'$ contains the atom $R(\vec z)$ that defines $\tau$,
    \item $Q'$ does not contain $R(\vec z)$ and is not Boolean, and
    \item $Q'$ is Boolean.
\end{enumerate}
In the first and second case, the root variable $x$ considered in step \ref{alg:choice-of-root} of the algorithm is a free variable, so $ $ and $ $ are disjoint sets of facts. Because of this, the parameters we consider add up: for example, if $Q'(E_1 \cup D_1\exo)$ contains $\ell_=^{(1)}$ facts $\vec t$ with $\tau(\vec t) = a$ and $Q'(E_2 \cup D_2\exo)$ contains $\ell_=^{(2)}$ such facts, then $Q'(E_1 \cup E_2 \cup D_1\exo \cup D_2\exo)$ contains $\ell_=^{(1)} + \ell_=^{(2)}$ such facts. The same holds for the over parameters we keep track of. To shorten notation, we write $\vec \ell = (\ell_<, \ell_=, \ell_>)$ and $(a,k,\vec \ell)$ for the quintuple $(a,k,\ell_<,\ell_=,\ell_>)$. With that, we obtain
\begin{align*}
    \combineU&(\P[Q',D_1], \P[Q',D_2])(a,k,\vec \ell) \\ &= \sum_{k_1=0}^{\size{D_1\endo}} \sum_{0\leq \aggsum(\vec \ell_1 )\leq\size{Q'(D_1)}} \P[Q',D_1](a, k_1, \vec \ell_1) \cdot \P[Q',D_2](a, k - k_1, \vec \ell - \vec \ell_1) \text.
\end{align*}
in the first and
\begin{align*}
    \combineU&(\P[Q',D_1], \P[Q',D_2])(k,\ell)  = \sum_{k_1 = 0}^{\size{D_1\endo}} \sum_{\ell_1 = 0}^{\size{Q'(D_1)}} \P[Q',D_1](k_1, \ell) \cdot  \P[Q',D_2](k - k_1, \ell - \ell_1)
\end{align*}
in the second case. In the third case, $Q'$ is Boolean, so we simply check whether the query is false for $\ell = 0$ or true for $\ell = 1$. Hence, we obtain
\begin{align*}
    \combineU&(\P[Q',D_1], \P[Q',D_2])(k,0)  = \sum_{k_1 = 0}^{\size{D_1\endo}} \P[Q',D_1](k_1, 0) \cdot  \P[Q',D_2](k - k_1, 0)
\end{align*}
and 
$\combineU(\P[Q',D_1], \P[Q',D_2])(k,1)$ is the number of the other $k$-subsets of $D_1\endo \cup D_2\endo$, which is $\binom{\size{D_1\endo} + \size{D_2\endo}}{k} - \combineU(\P[Q',D_1], \P[Q',D_2])(k,0).$

The formulas for $\combineX(\P[Q_1, D_1], \P[Q_2, D_2])$ are again straightforward. Since we consider the Cartesian of $Q_1$ and $Q_2$, all the parameters we keep track of (except $k$) behave multiplicatively and we need to go over all ways to obtain a given value for this product. For example, if $D_1$ contains $R$ and $D_2$ is not Boolean (all other combinations follow the same pattern and are left to the reader), we obtain the following formula for $\vec \ell \neq \vec 0$.
\begin{align*}
\combineX&(\P[Q_1, D_1], \P[Q_2, D_2])(a,k,\vec \ell) \\ &=\sum_{k_1= 0}^{\size{D_1\endo}} \; \sum_{\ell_2 \,\vert \gcd(\vec \ell)} \P[Q_1, D_1](a, k_1, \frac{1}{\ell_2}\vec \ell) \cdot \P[Q_2, D_2](k-k_1, \ell_2) \, ,    
\end{align*}
where $\ell_2 \,\vert \gcd(\vec \ell)$ means that $\ell_2$ divides $\gcd(\vec \ell)$ (or, equivalently, all values of $\vec \ell$). For $\vec \ell = \vec 0$, one of the factors must be $0$ or $\vec 0$, so we obtain
\begin{align*}
\combineX&(\P[Q_1, D_1], \P[Q_2, D_2])(a,k,\vec 0) \\ &=\sum_{k_1= 0}^{\size{D_1\endo}} \bigg( \sum_{\ell_2 = 1}^{\size{Q_2(D_2)}} \P[Q_1, D_1](a, k_1, \vec 0) \cdot \P[Q_2, D_2](k-k_1, \ell_2) \\
& \qquad \qquad + \sum_{0 \leq \aggsum(\vec \ell_1) \leq \size{Q_1(D_1)}}\P[Q_1, D_1](a, k_1, \vec \ell_1) \cdot \P[Q_2, D_2](k-k_1, 0) \bigg) 
\,.    
\end{align*}

This concludes the description of the instantiation for $\aggavg$ and $\aggquantile_q$.

\subsection{General Reduction}
In this section, we prove the following lemma, giving a general reduction from an \AggCQ built upon the CQ $\Qxyy$ to an \AggCQ with a CQ that is \allhierarchical but not \qhierarchical, where both use the same aggregation function.

\lemmageneralreductionnonqhierarchical*
The lemma follows directly from \Cref{item:generalization-lemma-item} of the next lemma, namely \Cref{lemma:generalization-hardness-cq}, using the same $h$ and $\tau_0$.

\begin{lemma}\label{lemma:generalization-hardness-cq}
Consider an \AggCQ of the form $\alpha\circ\tau\circ\Qxyy$. Let $Q_0$ be a CQ without self-joins, such that $Q_0$ is \allhierarchical but not q-hierarchical. There is a value function $\tau_0$, having the form $\tau \circ \tauidi{i}$, and a polynomial-time algorithm that, given an input database $D$ for $\Qxyy$, constructs an input database $D_0$ for $Q_0$, a function $h:D\endo\rightarrow D_0\endo$, and a function $g:\Qxyy(D)\rightarrow Q_0(D)$ such that:
\begin{enumerate}
    \item $h$ and $g$ are both bijections;
    \item For every $E\subseteq D\endo$ it holds that:
    \label{item:generalization-lemma-item}
    \begin{itemize}
        \item $g(\Qxyy(D\exo\cup E)) = Q_0(D_0\exo\cup h(E))$
    where $h(E)\defeq\set{h(f)\mid f\in E}$ and, for every set $T$ of tuples, $g(T)\defeq\set{g(t)\mid t\in T}$;
        \item $\multiset{\tau(t)\mid t\in \Qxyy(D\exo\cup E)}=\multiset{\tau_0(t_0)\mid t_0\in Q_0(D_0\exo\cup h(E))}$.
        \end{itemize}
\end{enumerate}
\end{lemma}
\begin{proof}
Recall that $\Qxyy$ is defined by $\Qxyy(x) \la R(x,y), S(y)$. Let $x_0$ and $y_0$ be two variables of $Q_0$ such that $\atoms(x_0)\subsetneq\atoms(y_0)$ where $x_0$ is a head variable and $y_0$ is an existential variable. Let $\varphi^R_0$ be an atom that contains both $x_0$ and $y_0$, and let $\varphi^S_0$ be an atom that contains $y_0$ but not $x_0$.  Let $R_0$ and $S_0$ be the relation names of $\varphi^R_0$ and $\varphi^S_0$, respectively. Choose an arbitrary constant $c\in\consts$. 

Let $D$ be an input database for $\Qxyy$. Let $a$ and $b$ be two values in $D$, and let $\varphi_0$ be an atom of $Q_0$. We denote by
$\varphi_0[a,b]$ the fact that is obtained from $\varphi_0$ by replacing $x_0$ with $a$, replacing $y_0$ with $b$, and replacing every other variable with $c$. If $\varphi_0$ contains only $x_0$, then we denote by $\varphi_0[a]$ the fact $\varphi_0[a,b]$ for some arbitrary $b$ (which makes no difference). 

The database $D_0$ consists of the following facts:
\begin{itemize}
    \item $\varphi_0[a,b]$ for every tuple tuple $(a,b)\in\Qxyy(D)$;
    \item $\varphi^R_0[a,b]$ for every fact $R(a,b)\in D$;
    \item $\varphi^S_0[a]$ for every fact $S(a)\in D$.
\end{itemize}
We define every fact $\varphi^R_0[a,b]$ as endogenous if $R(a,b)$ is endogenous, and every fact $\varphi^S_0[a]$ as endogenous if $S(a)$ is endogenous. Every other fact of $D_0$ is exogenous. 

The function $h$ maps every endogenous $R(a,b)$ in $D\endo$
to the endogenous $\varphi^R_0[a,b]$ in $D_0\endo$, and every endogenous $S(a)$ in $D\endo$ to the endogenous $\varphi^S_0[a]$ in $D_0\endo$.

Recall that $Q_0$ has no self-joins. Therefore, every answer in $Q_0(D_0)$ is obtained from an answer $(a)\in \Qxyy(D)$ by taking the sequence of variables in the head of $Q_0$, replacing $x_0$ with $a$, and replacing every other variable with $c$; we denote this tuple by $t_0^a$. Note also that every answer $(a)\in \Qxyy(D)$ gives rise to the answer $t_0^a$ in $Q_0(D_0)$. Hence, we define $g((a))=t_0^a$. 

Recall that $\tau$ is defined over tuples of length one. Then $\tau_0$ is defined by $\tau_0\defeq \tau \circ \tauidi{i}$, that is, application of $\tau$ to the $i$th entry of the tuple.

The items of the lemma then follow directly from the construction of $D_0$, $h$, $g$, and $\tau_0$.
\end{proof}

\subsection{Hardness for the Simplest non-Q-hierarchical CQ}
In this section, we prove the following.

\lemmaavgquantilesimplesthard*

\subsubsection{Average}
We begin with the $\aggavg$ function.

\begin{lemma}
\label{lemma:avg-non-q-hierarchical-hard}
Let $Q$ be a CQ without self-joins, such that $Q$ is \allhierarchical but not q-hierarchical. There is a localized value function $\tau$ such that for 
$A=\aggavg\circ\tau\circ Q$, computing $\Shapley(f, A)$, given a database $D$ and a fact $f$, is \fpsharpp-complete. 
\end{lemma}
To prove \Cref{lemma:avg-non-q-hierarchical-hard}, we first prove hardness for the query $A=\aggavg\circ\taurelu\circ \Qxyy$. %
To prove \Cref{lemma:avg-non-q-hierarchical-hard}, we can focus on the query $A$ and apply \Cref{lemma:general-reduction-non-q-hierarchical}; thus, it suffices to prove that:
\begin{lemma}\label{lemma:avg-simplest-hard}
For $A=\aggavg\circ\taurelu\circ \Qxyy$, computing $\Shapley(f, A)$, given a database $D$ and a fact $f$, is \fpsharpp-complete. 
\end{lemma}
\begin{proof}
We will show a reduction from the \e{\#Set-Cover} problem: given a set $\X$ of elements, a collection $\Y$ of nonempty subsets of $\X$, determine the number of collections $\C\subseteq\Y$ that cover all of $\X$ (i.e., $\X=\cup\C$). This problem is known to be \#P-complete~\cite{DBLP:journals/siamcomp/ProvanB83}. 
In the remainder of this section, we fix the input $(\X,\Y)$ for the problem, where $\X\defeq\set{1,\dots,n}$ and $\Y\defeq\set{Y_1,\dots,Y_m}$. 
\subparagraph{Construction.}
For $i\in\set{0,\dots,n}$ and $j\in\set{0,\dots,m}$, we denote by $Z_{i,j}$ the number of subsets of $\Y$ with $j$ members, covering precisely $i$ elements from $\X$. In notation:
\begin{equation}
    Z_{i,j}\defeq |\set{\C\subseteq\Y : \lvert\cup\C\rvert=i\mbox{ and }|\C|=j}|
\end{equation}
For example, $Z_{0,j}=0$ whenever $j>0$, and $Z_{0,0}=1$. Importantly, we can get the number of set covers by $(\X,\Y)$ by summing up $Z_{n,j}$ over $j=0,\dots,m$. We will show how to compute the quantities $Z_{i,j}$ using the Shapley value.

We construct a collection of $n\cdot m$ databases $D_{q,r}$, where $q\in\set{0,\dots,n}$ and $r\in\set{0,\dots,m}$. \Cref{fig:avg-reduction} shows an illustration of the construction.

The database $D_{q,r}$ consists of the following facts:
\begin{itemize}
\item $R$ contains the following exogenous facts:
\begin{itemize}
\item $R(-i,j)$ for every $i\in\X$ and $Y_j\in\Y$ such that $i\in Y_j$;
\item $R(-n-i,m+1)$ for $i=1,\dots,q+1$;
\item $R(1,m+1+j)$ for $j=1,\dots,r$;
\item $R(1,0)$.
\end{itemize}
\item $S$ contains the following endogenous facts:
\begin{itemize}
    \item $S(j)$ for $j=1,\dots,m$;
    \item $S(m+1+j)$ for $j=1,\dots,r$;
    \item $S(0)$.
\end{itemize}
and the following exogenous fact:
\begin{itemize}
    \item $S(m+1)$.
\end{itemize}
\end{itemize}
We denote by $f$ the fact $S(0)$. For each database $D_{q,r}$, we denote by $\Shapley_{q,r}(f, A)$ the Shapley value of $f$ for the query $A$ over $D_{q,r}$. We will show how to compute the set of values $Z_{i,j}$ from the set of values $\Shapley_{q,r}(f, A)$.

\begin{figure}
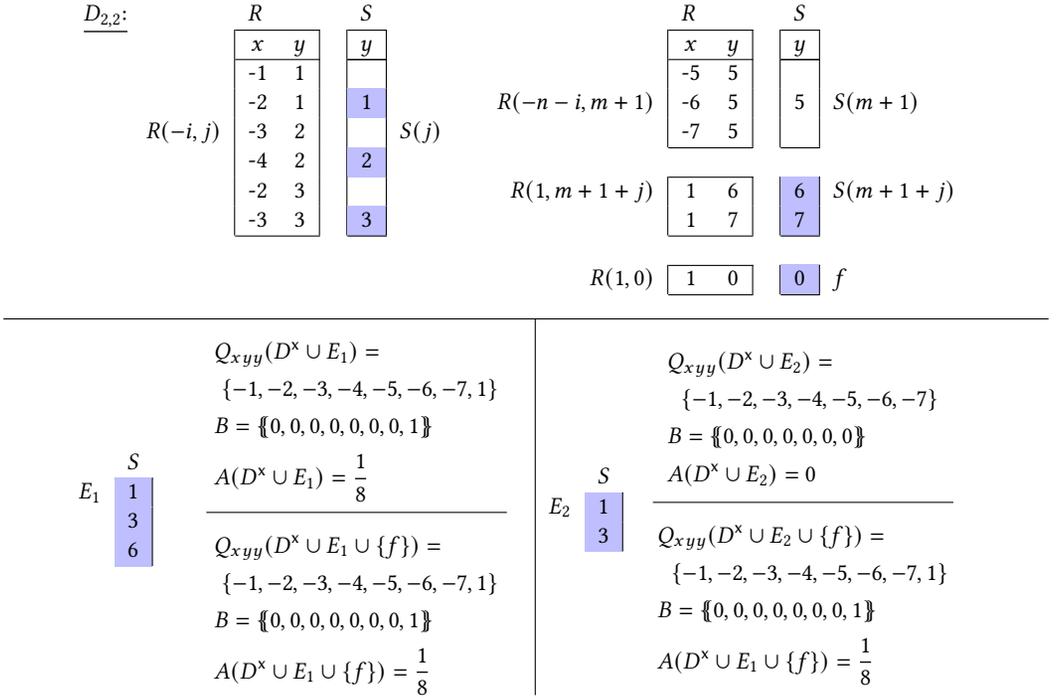

\small
\newcolumntype{E}{>{\cellcolor{blue!25}}c|}
    \centering
\underline{$D_{2,2}$:}
\begin{tabular}[t]{r|cc|c|c|l}
\multicolumn{1}{l}{} & \multicolumn{2}{l}{$R$} & \multicolumn{1}{l}{} & \multicolumn{1}{l}{$S$} & \\
\cline{2-3}\cline{5-5}
& $x$ & $y$ & & $y$ & \\
\cline{2-3}\cline{5-5}\cline{5-5}
& -1 & 1 & &  &\\
& -2 & 1 & &  \multicolumn{1}{E}{1} &\\
$R(-i,j)$ & -3 & 2 & &  & $S(j)$\\
& -4 & 2 & & \multicolumn{1}{E}{2}   &\\
& -2 & 3 & & &\\
& -3 & 3 & & \multicolumn{1}{E}{3} \\
\cline{2-3}\cline{5-5}
\end{tabular}
\quad
\begin{tabular}[t]{r|cc|c|c|l}
\multicolumn{1}{l}{} & \multicolumn{2}{l}{$R$} & \multicolumn{1}{l}{} & \multicolumn{1}{l}{$S$} & \\
\cline{2-3}\cline{5-5}
& $x$ & $y$ & & $y$ & \\
\cline{2-3}\cline{5-5}\cline{5-5}
& -5 & 5 & &  &\\
$R(-n-i,m+1)$ & -6 & 5 & & 5 & $S(m+1)$ \\
& -7 & 5 & &  &
\\
\cline{2-3}\cline{5-5}
\multicolumn{1}{c}{}\\
\cline{2-3}\cline{5-5}
$R(1,m+1+j)$ & 1 & 6 & &  \multicolumn{1}{E}{6} & $S(m+1+j)$\\
& 1 & 7 & &  \multicolumn{1}{E}{7} &\\
\cline{2-3}\cline{5-5}
\multicolumn{1}{c}{}\\
\cline{2-3}\cline{5-5}
$R(1,0)$ & 1 & 0 & &  \multicolumn{1}{E}{0} & $f$ \\
\cline{2-3}\cline{5-5}
\end{tabular}
\vskip1em
\hrule
\begin{minipage}{6cm}
\small
\begin{tabular}{rc}
& $S$ \\ \cline{2-2}
$E_1$ & \multicolumn{1}{E}{1}\\
& \multicolumn{1}{E}{3}\\
& \multicolumn{1}{E}{6}\\
\end{tabular}
\quad\quad
\begin{minipage}{4cm}
\begin{align*}
    &\Qxyy(D\exo\cup E_1)=\\
    &\;\set{-1,-2,-3,-4,-5,-6,-7,1}\\
    & B = \multiset{0,0,0,0,0,0,0,1}\\
    &A(D\exo\cup E_1)=\frac{1}{8}
\end{align*}
\hrule
\begin{align*}
    &\Qxyy(D\exo\cup E_1\cup\set{f})=\\
    &\;\set{-1,-2,-3,-4,-5,-6,-7,1}\\
    & B = \multiset{0,0,0,0,0,0,0,1}\\
    & A(D\exo\cup E_1\cup\set{f})=\frac{1}{8}
\end{align*}
\end{minipage}
\end{minipage}
\;
\vrule
\begin{minipage}{6cm}
\small
\begin{tabular}{rc}
& $S$ \\ \cline{2-2}
$E_2$ & \multicolumn{1}{E}{1}\\
& \multicolumn{1}{E}{3}\\
\end{tabular}
\quad
\begin{minipage}{4cm}
\begin{align*}
    &\Qxyy(D\exo\cup E_2)=\\
    &\;\;\set{-1,-2,-3,-4,-5,-6,-7}\\
    & B = \multiset{0,0,0,0,0,0,0}\\
    & A(D\exo\cup E_2)=0
\end{align*}
\hrule
\begin{align*}
    &\Qxyy(D\exo\cup E_2\cup\set{f})=\\
    &\;\;\set{-1,-2,-3,-4,-5,-6,-7,1}\\
    & B = \multiset{0,0,0,0,0,0,0,1}\\
    & A(D\exo\cup E_1\cup\set{f})=\frac{1}{8}
\end{align*}
\end{minipage}
\end{minipage}
\caption{Example of the reduction from \#Set-Cover for $\X=\set{1,2,3,4}$ (and $n=4$) and $\Y=\set{\set{1,2},\set{3,4},\set{2,3}}$ (and $m=3$). Endogenous tuples have a purple shade; other tuples are all exogenous. The figure shows $D_{2,2}$ on the top, and examples of subsets $E$ of $D_{2,2}\endo$ (bottom), along with the result of $A$ over the corresponding databases $D_{2,2}\exo\cup E$. The figure denotes by $B$ the bag of values $\tau(a)$ over all answers $a$. }
    \label{fig:avg-reduction}
\Description{Example of the reduction from \#Set-Cover}
\end{figure}

\subparagraph*{Formula for the Shapley value.}
Fix a database $D_{q,r}$. Let us consider the Shapley game, and see how much the tuple $t=S(0)$ contributes to the overall utility when it is added.
\begin{itemize}
\item If any tuple $S(m+1+j)$ is selected before $t$, then adding $t$ does not change the set of answers and, hence, its contribution is zero.
\item If no tuple $S(m+1+j)$ has been selected before $t$ (for $j=1,\dots,r$), then the average before the addition of $t$ is zero, and after the addition of $t$ it is $1/(i+q+1)$, where $i$ is total the number of elements in the subsets $Y_{j'}$ selected before $t$ via the tuple $S(j')$.
\end{itemize}
Hence, the Shapley value of $t$ is the following:
\begin{equation}
\label{eq:shapley-average-reduction}
\Shapley_{q,r}(f, A) = 
\sum_{j=0}^m\sum_{i=0}^{n}
\frac{j!\cdot (m+r-j)!}{(m+r+1)!}
\cdot
\frac{Z_{i,j}}{i+q+1}
\end{equation}

\subparagraph*{Solving the linear system.}
Suppose that we have an oracle for computing each $\Shapley_{q,r}(f, A)$. 
From \Cref{eq:shapley-average-reduction} we get a system of linear equations of the form $\mathbf{L}\vec Z=\vec a$, where $\mathbf{L}$ is an $(\ell\times\ell)$-matrix for $\ell=(m+1)\cdot(n+1)$ and $\vec a$ is an $\ell$-vector. Moreover, we can compute every element of $\mathbf{L}$ and $\vec a$. Hence, to restore the $Z_{i,j}$, we need to show that $\mathbf{L}$ is invertible. We do so in the remainder of this section.

Let us denote by $\mathbf{M}$ the $(m+1)\times(m+1)$-matrix defined by
\[\mathbf{M}_{r',j'}\defeq \frac{(j'-1)!\cdot (m+r'-j')!}{(m+r')!}\,.\]
and $\mathbf{N}$ the $(n+1)\times(n+1)$-matrix given by
\[\mathbf{N}_{q',i'}\defeq \frac{1}{q'+i'-1}\,.\]
For $j',q'\in\set{1,\dots,m+1}$ and $i',q'\in\set{1,\dots,n+1}$, the element
of $\mathbf{L}$ 
at the row corresponding to $(j',q')$ and columns corresponding to $(i',r')$ is 
$\mathbf{M}_{r',j'}\cdot \mathbf{N}_{q',i'}$. Hence, $\mathbf{L}$ is the \e{Kronecker product}\footnote{The Kronecker product has also been used by Amarilli and Kimelfeld~\cite{DBLP:journals/lmcs/AmarilliK22} for proving \fpsharpp-hardness of a related problem in databases, namely \e{uniform reliability}.} $\mathbf{M}\otimes \mathbf{N}$. Recall if $\mathbf{C}$ is an $(\ell\times\ell)$-matrix and $\mathbf{C'}$ is an $(\ell'\times\ell')$-matrix,
the Kronecker product $\mathbf{C}\otimes \mathbf{C'}$ is the $(\ell\cdot\ell'\times\ell\cdot\ell')$-matrix given as follows:
$$
\begin{pmatrix}
\mathbf{C}_{1,1}\mathbf{C'} & 
\mathbf{C}_{1,2}\mathbf{C'} &
\cdots & \mathbf{C}_{1,\ell}\mathbf{C'}\\
\mathbf{C}_{2,1}\mathbf{C'} & 
\mathbf{C}_{2,2}\mathbf{C'} &
\cdots & \mathbf{C}_{2,\ell}\mathbf{C'}\\
\vdots & \vdots & \cdots & \vdots\\
\mathbf{C}_{\ell,1}\mathbf{C'} & 
\mathbf{C}_{\ell,2}\mathbf{C'} &
\cdots & \mathbf{C}_{\ell,\ell}\mathbf{C'}\\
\end{pmatrix}	
$$
It is known that $\mathbf{C}\otimes \mathbf{C'}$ is invertible if and only if both $\mathbf{C}$ and $\mathbf{C'}$ are invertible (see, e.g.,~\cite{Henderson01101983}). So, it is left to prove that $\mathbf{M}$ and $\mathbf{N}$ are both invertible.
The matrix $\mathbf{M}$ can be converted, by column and row multiplications, to the Hankel matrix  $\mathbf{M'}$ defined by $\mathbf{M'}_{r',j'}=(r'+j'-1)!$ for $r',j'\in\set{1,\dots,n}$. 
Specifically, we multiply every row $r'$ by the number $(m+r')!$, divide every column $j'$ by $(j'-1)!$, and reverse the order of the columns $j'$.
The matrix
$\mathbf{M'}$ is known to be invertible (e.g.,~\cite{determinants2002}).
The matrix $\mathbf{N}$ is known as the \e{Hilbert matrix} (which is also a Hankel matrix), and is also known to be invertible (e.g.,~\cite{choi1983tricks}). 
\end{proof}

\subsubsection{Quantile}
We now proceed to the $q-\aggquantile$ function. We prove the following lemma.
\begin{lemma}\label{lemma:qnt-simplest-hard}
For fixed $q \in (0,1)_{\mathbb Q}$ and $A=\aggquantile_q\circ\taugb{0}\circ \Qxyy$, computing $\Shapley(f, A)$, given a database $D$ and a fact $f$, is \fpsharpp-complete. 
\end{lemma}

We show this lemma using a reduction from a cooperative game associated with the \e{\#Set-Cover} problem for which Shapley value computation is known to be \fpsharpp-complete \cite{gilad2024importanceparametersdatabasequeries}.
\begin{lemma}[\protect{\cite[Proof~of~Theorem~6.3]{gilad2024importanceparametersdatabasequeries}}]
    For an input $(\X,\Y)$ to the \e{\#Set-Cover} problem with $\Y = \set{Y_1,\dots,Y_m}$, consider the cooperative game with player set $P = \set{1, \ldots, m}$ and utility 
    \[
    \nusetcover(C) = \begin{cases}
        1 & \text{if } \bigcup_{i \in C} = \X \\
        0 & \text{otherwise.}
    \end{cases}
    \]
    Then, computing $\Shapley(i,\nusetcover)$ for $(\X,\Y)$ and an index $i$ of $\Y$ is \fpsharpp-complete
\end{lemma}

\begin{proof}[Proof of \Cref{lemma:qnt-simplest-hard}]
    Let $\X = \set{1, \ldots, n}$ and $\Y = \set{Y_1, \ldots, Y_m}$ be an input to the \e{\#Set-Cover} problem. We aim to construct a database $D$ such with $m$ endogenous facts $f_1, \ldots, f_m$ such that $\Shapley(f_i, A) = \Shapley(i, \nusetcover)$. Let $q = \frac{a}{b}$. Then $D$ consist of the following facts:
    \begin{itemize}
        \item $R$ consists of the following exogenous facts:
        \begin{itemize}
            \item $R(j\cdot b\cdot (b-a) -\ell, i)$ for each $j \in Y_i$ and each $\ell = 0, 1, \ldots, b \cdot (b-a) - 1$,
            \item $R( - \ell, 0)$ for $\ell = 1, 2, \ldots, b\cdot a \cdot n$, and
            \item $R( n \cdot b \cdot (b - a) + 1, 0)$.
        \end{itemize}
        \item $S$ consists of the following endogenous facts:
        \begin{itemize}
            \item $S(i)$ for $i = 1, \ldots, m$ 
        \end{itemize}
        and the following exogenous fact:
        \begin{itemize}
            \item $S(0)$.
        \end{itemize}
    \end{itemize}
    
    Now, let $C \subseteq \set{1, \ldots, m}$ and $C' = \set{S(i) \, \mid \, i \in C} \subseteq D\endo$. If $\size{\bigcup_{i \in C} Y_i} = t$, then $Q(C' \cup D\exo)$ contains the negative numbers $-1, \ldots, -b\cdot a \cdot n$ as well as $t \cdot b \cdot (b-a) + 1$ positive numbers: it contains $n \cdot b \cdot (b-a) + 1$ and for each number $c \in \bigcup_{i \in C}Y_i$ the numbers from $(c - 1) \cdot b \cdot (b-a) + 1$ to $c \cdot b \cdot (b-a)$. Hence $(\taugb{0} \circ Q)(C' \cup D\exo)$ is the bag over $\set{0,1}$ with $\mu(0) = b \cdot a \cdot n$ and $\mu(1) = t \cdot b \cdot (b-a) + 1$.

    Since $(\taugb{0} \circ Q)(C' \cup D\exo)$ has size $b \cdot ( a \cdot n + t \cdot (b-a)) + 1$, its $\frac{a}{b}$-quantile is the $a \cdot (a \cdot n + t \cdot (b - a)) + 1$ smallest element of $(\taugb{0} \circ Q)(C' \cup D\exo)$. Now, if $t < n$, then $a \cdot (a \cdot n + t \cdot (b - a)) + 1 \leq a \cdot (a \cdot n + n \cdot (b - a)) = a \cdot b \cdot n$, so $A(C' \cup D\exo) = 0$. Otherwise, $t = n$ and $a \cdot (a \cdot n + t \cdot (b - a)) + 1 = b \cdot a \cdot n + 1$, so $A(C' \cup D\exo) = 1$.

    This shows that the map $i \mapsto S(i)$ is a utility-preserving bijection between the cooperative games $(\set{1, \ldots, m}, \nusetcover)$ and $(D\endo, A)$, so $\Shapley(S(i), A) = \Shapley(i, \nusetcover)$. 
\end{proof}

\section{Omitted Proofs from \Cref{sec:duplicates}}
In this section, we prove the theorem of \Cref{sec:duplicates}:

\thmduplicatesclassification*

\subsection{Hardness}

\begin{figure}
\def\descall{\Cref{lemma:dup-Qxyy-both-hard,lemma:dup-non-q-hierarchical-hard}: reduction from $\aggduplicates\circ\taurelu\circ \Qxyy$}
\def\descq{\Cref{lemma:dup-Qxyy-both-hard,lemma:dup-non-sq-hierarchical-hard}: reduction from $\aggduplicates\circ\tauidi{1}\circ\QFxyy$}
    \centering
    \input{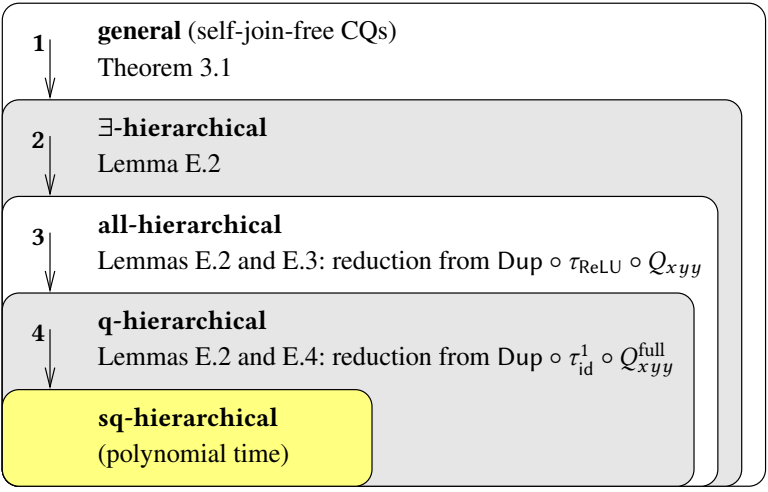}
    \caption{Structure of the hardness proof for Shapley value computation for  $\aggduplicates$-aggregation}
    \label{fig:dup-proof}
    \Description{Structure of the hardness proof for Shapley value computation for has-duplicates aggregation}
\end{figure}
The proof of hardness is illustrated in \Cref{fig:dup-proof}. Starting from the class of all CQs without self-joins, in every step we show hardness that eliminates a smaller class of CQs, until we get to the class of \sqhierarchical CQs.

\subsection*{Steps 1 and 2}

The non-\exhierarchical CQs are covered by \Cref{thm:livshits-hardness}.
Note, however, that
\Cref{thm:hardness-const-per-single} does not apply to the aggregate function has-duplicates, since it is not constant per singleton (since, e.g., $\aggduplicates(\multiset{1,1})\neq\aggduplicates(\multiset{1})$). The following result extends \Cref{thm:hardness-const-per-single} to this aggregate function. 

\begin{lemma}\label{lemma:dup-hard-non-all}
Let $A=\aggduplicates\circ\tau\circ Q$ be an \AggCQ such that $\tau$ is a constant function. Suppose that $Q$ has no self-joins. If $Q$ is not \allhierarchical, then it is \fpsharpp-hard to compute
$\Shapley(f, A)$.    
\end{lemma}
\begin{proof}
If $Q$ is not \exhierarchical, then we get hardness from \Cref{thm:livshits-hardness}. So, in the remainder of this proof, we assume that $Q$ is \exhierarchical but not \allhierarchical. Suppose that $Q$ is a Cartesian product of $Q_1$, $Q_2$, \dots, $Q_k$, where each $Q_i$ is connected. Since $Q$ is not \allhierarchical, at least one of $Q_i$ is not \allhierarchical, say $Q_1$ (without loss of generality). Note that $Q_1$ cannot be Boolean, since we assume that $Q$ is \exhierarchical (hence, each Boolean connected component is hierarchical). 

Let $Q_1^b$ be the Boolean CQ obtained from $Q_1$ by making its free variables existential. Then, from prior work~\cite{DBLP:journals/lmcs/LivshitsBKS21} we know that $\Shapley(f,Q_1^b)$ is \fpsharpp-hard to compute. We will show a reduction from $\Shapley(f,Q_1^b)$ to $\Shapley(f,A)$. 

Let $D_1$ be an input database for $Q_1^b$ and $f\in D_1\endo$.
We construct a database $D$ for $A$ as follows. Let $c\in\consts$ be a value that does not appear in $D_1$. We begin with $D=D_1$ and add to $D$ a single exogenous fact $f_\varphi$ for each atom $\varphi$ of $Q$; the fact $f_\varphi$ is obtained from $\varphi$ by replacing every variable with $c$. Note that $D\endo=D_1\endo$ and, moreover, $D\exo$ is the disjoint union of $D_1\exo$ and the set of new facts that we have added.

Let $\vec t_c$ be the tuple that is obtained from the head of $Q$ by replacing every variable with $c$. For each answer $\vec t\in Q_1(D_1)$, let $\vec t[c]$ be the tuple extends $\vec t$ to an answer in $Q(D)$ by replacing every variable that is not in $Q_1$ with $c$. 

Observe the following. For each subset $E$ of $D\endo$ it holds that 
$$
Q(E\cup D\exo) = \set{\vec t[c]\mid \vec t\in Q_1(E\cup D_1\exo)} \cupdot \set{\vec t_c}
$$
In particular, $Q_1(E\cup D_1\exo)$ is nonempty (i.e., $Q_1^b(E\cup D_1\exo)$ is true) if and only if $Q(E\cup D\exo)$ contains at least two tuples (and otherwise it contains exactly one tuple). Since $\tau$ is a constant function, we conclude that $Q_1^b(E\cup D_1\exo)$ is true if and only if $A(E\cup D\exo)=1$. Hence,
$\Shapley(f,Q_1^b)=\Shapley(f,A)$, which completes the reduction.
\end{proof}

\subsection*{Steps 3 and 4}
We continue with the hardness of two specific CQs. 

\begin{lemma}
\label{lemma:dup-Qxyy-both-hard}
Computing $\Shapley(f, A)$, given a database $D$ and a fact $f$, is \fpsharpp-complete for aggregate queries $A$:
\begin{enumerate}
    \item $\aggduplicates\circ\taurelu\circ \Qxyy$
    \item $\aggduplicates\circ\tauidi{1}\circ\QFxyy$  for $\QFxyy(x,y) \la R(x,y), S(y)$
\end{enumerate}
\end{lemma}
\begin{proof}
We first discuss $\aggduplicates\circ\taurelu\circ \Qxyy$, and then adjust the proof to the second query.
We construct a reduction from the problem of computing the permanent of a 0/1-matrix, which is known to be \#P-hard due to Valiant's Theorem~\cite{DBLP:journals/tcs/Valiant79}. We define this problem as a covering problem, using the same notation as in \Cref{lemma:avg-non-q-hierarchical-hard}: given a set $\X$ of elements, a collection $\Y$ of pairs of elements of $\X$, determine the number of \e{exact covers}, that is, pairwise-disjoint collections $\C\subseteq\Y$ that cover all of $\X$. In the remainder of this section, we fix the input $(\X,\Y)$ for the problem, where $\X\defeq\set{1,\dots,n}$ and $\Y\defeq\set{Y_1,\dots,Y_m}$. We can assume that $n$ is even, or otherwise there is no exact cover. 

For $j\in\set{0,\dots,m}$, we denote by $Z_{j}$ the number of subsets of $\Y$, consisting of precisely $j$ pairwise-disjoint members. Note that the permanent is the number $Z_{n/2}$. We will show how to compute the quantities $Z_{j}$ using the Shapley value for $A$.

We construct a collection of databases $D_{r}$ for $r\in\set{0,\dots,m}$. The database $D_{r}$ consists of the following facts:
\begin{itemize}
\item $R$ contains the following exogenous facts:
\begin{itemize}
\item $R(i,j)$ for every $i\in\X$ and $Y_j\in\Y$ such that $i\in Y_j$;
\item $R(0,0)$;
\item $R(-1,-1)$;
\item $R(-2,m+r')$ for $r'=1,\dots,r$.
\end{itemize}
\item $S$ contains the exogenous fact
\begin{itemize}
    \item $S(-1)$;
\end{itemize}
and the following endogenous facts:
\begin{itemize}
    \item $S(0), S(1), \dots, S(m)$;
    \item $ S(m+r')$ for $r'=1,\dots,r$.
\end{itemize}
\end{itemize}
We denote by $f$ the fact $S(0)$. For each database $D_{r}$, we denote by $\Shapley_{r}(f, A)$ the Shapley value of $f$ for the query $A$ over $D_{r}$. We will show how to compute the set of values $Z_{j}$ from the set of values $\Shapley_{r}(f, A)$.

Fix a database $D_{r}$. Let us consider the Shapley game, and see how much the tuple $t=S(0)$ contributes to the overall utility when it is added to a subset $D'$.
\begin{itemize}
\item If any tuple $S(m+r')$ is selected before $t$, then the set of values $\tau(t)$ includes a duplicate (hence, $A(D')=1$) and adding $t$ does not change the answer.
\item If the collection of subsets $Y_i$ with $S(i)\in D$ includes two or more intersecting pairs, then $A(D')=1$ and adding $t$ does not change the answer.
\item If no tuple $S(m+r')$ has been selected before $t$ (for $r'=1,\dots,r$) and the collection of subsets $Y_i$ with $S(i)\in D$ is pairwise disjoint, then $A(D')=0$ before $t$ is added and $A(D')=1$ after it is added. Hence, the contribution is one.
\end{itemize}
Therefore, the Shapley value of $t$ is the following:
\begin{equation}
\label{eq:shapley-average-reduction-in-proof}
\Shapley_{r}(f, A) = 
\sum_{j=0}^m
\frac{j!\cdot (m+r-j)!}{(m+r+1)!}
\cdot
Z_{j}
\end{equation}
From here on, we continue similarly to the proof of \Cref{lemma:avg-non-q-hierarchical-hard}, showing that the $Z_j$ are the solution of the full-rank matrix $\mathbf{M'}$ defined by $\mathbf{M'}_{r',j'}=(r'+j'-1)!$ for $r',j'\in\set{1,\dots,n}$. 

This completes the proof for $\aggduplicates\circ\taurelu\circ \Qxyy$. For $\aggduplicates\circ\tauidi{1}\circ\QFxyy$, we use the exact same proof, with the following change (to accommodate the change from $\taurelu$ to $\tauidi{2}$): In the relation $R$, we replace every fact $R(a,b)$ with $a<0$ with $R(0,b)$. For example, $R^D$ contains $(0,-1)$ instead of $(-1,-1)$, and $(0,m+r')$ for $r'=1,\dots,r$.
\end{proof}

Combining \Cref{lemma:general-reduction-non-q-hierarchical} with \Cref{lemma:dup-Qxyy-both-hard} (concerning $\Qxyy$), we conclude the following lemma.
\begin{lemma}\label{lemma:dup-non-q-hierarchical-hard}
Let $Q$ be a CQ without self-joins, such that $Q$ is \allhierarchical but not q-hierarchical. There is a localized value function $\tau$ such that for 
$A=\aggduplicates\circ\tau\circ Q$, computing $\Shapley(f, A)$, given a database $D$ and a fact $f$, is \fpsharpp-complete. 
\end{lemma}

Using a similar proof technique, we can generalize the part of \Cref{lemma:dup-Qxyy-both-hard} concerning $\QFxyy$ as follows.
\begin{lemma}\label{lemma:dup-non-sq-hierarchical-hard}
Let $Q$ be a CQ without self-joins, such that $Q$ is \qhierarchical but not \sqhierarchical. There is a localized value function $\tau$ such that for 
$A=\aggduplicates\circ\tau\circ Q$, computing $\Shapley(f, A)$, given a database $D$ and a fact $f$, is \fpsharpp-complete. 
\end{lemma}
\begin{proof}
We prove the lemma by showing that the claim corresponding to 
\Cref{lemma:generalization-hardness-cq} holds here as well. Specifically, let $Q_0$ be a CQ without self-joins, such that $Q_0$ is \qhierarchical but not \sqhierarchical.  We need to show a value function $\tau_0$ and a polynomial-time algorithm that, given an input database $D$ for $\QFxyy$, constructs an input database $D_0$ for $Q_0$, a function $h:D\endo\rightarrow D_0\endo$, and a function $g:\Qxyy(D)\rightarrow Q_0(D)$, so that the two conditions of \Cref{lemma:generalization-hardness-cq} hold. The proof, given below, is similar to that of \Cref{lemma:generalization-hardness-cq}.

Recall that $\QFxyy$ is defined by $\QFxyy(x,y) \la R(x,y), S(y)$. Since $Q_0$ is not \sqhierarchical, we conclude that there are variables $x_0$ and $y_0$ such that $x_0$ is a free (head) variable and 
$\atoms(x_0)\subsetneq\atoms(y_0)$. From the fact that $Q_0$ is \qhierarchical we conclude that $x_0$ is a free variable as well. 
We define  $\varphi^R_0$, $\varphi^S_0$, $R_0$, $S_0$ and $c$ similarly to \Cref{lemma:generalization-hardness-cq}. Moreover, for an input database $D$ for $\QFxyy$, we construct the database $D_0$ and the bijection $h$ in the exact same way as in  \Cref{lemma:generalization-hardness-cq}.

Every answer in $Q_0(D_0)$ is obtained from an answer $(a,b)\in \QFxyy(D)$ by taking the sequence of variables in the head of $Q_0$, replacing $x_0$ with $a$, replacing $y_0$ with $b$, and replacing every other variable with $c$; we denote this tuple by $t_0^{a,b}$. We then define $g((a,b))=t_0^{a,b}$. 

Finally, the function $\tau_0$ is defined as follows. Let $i$ be a position in the head of $Q_0$ where $x_0$ occurs. In every tuple $t=t_0^{a,b}\in Q_0(D_0)$, the value $a$ is the element in the $i$th position; we denote this element by $t[i]$. Then $\tau_0$ is taken as $\taurelui{i}.$
\end{proof}

\subsection{Algorithm}
We now present a polynomial-time algorithm for computing $\sumword_k(A,D)$ for $A=\aggduplicates\circ\tau\circ Q$ where $\tau$ is localized and $Q$ is \sqhierarchical. As described in 
\Cref{sec:template}, this proves the tractability of $\Shapley(f, A)$ in this case.

\subsubsection{Computing $\P^0[Q,D]$ and $\P^1[Q,D]$}
\label{sec:P0-and-P1}
We begin by showing how to compute the data structure $\P$ defined by $$\P[Q,D]\defeq (\P^0[Q,D],\P^1[Q,D])$$ 
where $\P^i[Q,D]$ is the function that maps every number $k=0,\dots,|D\endo|$ to the number of $k$-subsets $E$ of $D\endo$ such that $|Q(D\exo\cup E)|=i$. 
In other words, $\P^0[Q,D](k)$ is the number of subsets of $D\endo$ where $Q$ has no answers, and $\P^1[Q,D](k)$ is the number of subsets of $D\endo$ where $Q$ has precisely one answer. 

The computation adopts and adjusts the generic algorithm presented in \Cref{fig:generic-alg}.
The subroutine $\combineX(\P[Q_1,D_1],\P[Q_2,D_2])$ is implemented as follows. We denote $m_i=|D_i\endo|$ for $i=1,2$. Then, denoting $D'=D_1\cup D_2$:
\begin{align*}
\P^0[Q,D_1\cup D_2](k) =  \sum_{\ell=0}^k &
\P^0[Q_1,D_1](\ell)\cdot \binom{m_2}{k-\ell} + 
\binom{m_1}{\ell}\cdot \P^0[Q_2,D_2] 
\\ &
- 
\P^0[Q_1,D_1](\ell)\cdot  \P^0[Q_2,D_2](k-\ell)
\end{align*}
(i.e., either $Q_1(D_1)$ is empty or $Q_2(D_2)$ is empty, and we subtract to account for adding twice the possibility that both are empty). In addition, if $Q'$ is a Cartesian product of $Q_1$ and $Q_2$:
$$\P^1[Q',D_1\cup D_2](k) = \sum_{\ell=0}^k \P^1[Q_1,D_1](\ell)\cdot\P^1[Q_2,D_2](k-\ell)$$
(i.e., both sides of the product should give precisely one answer).

For the $\combineU$ function, we use the assumption that $Q$ is q-hierarchical. This means that $Q$ is either Boolean or the answer sets $Q(D_1)$ and $Q(D_2)$ are disjoint since the vertical split is applied to a free variable. In the first case, we adopt the algorithm for the corresponding task by Livshits et al.~\cite{DBLP:journals/lmcs/LivshitsBKS21}. In the second case, we have the following:
\begin{align*}
\P^1[Q',D_1\cup D_2](k) = \sum_{\ell=0}^k 
&
\P^1[Q',D_1](\ell)\cdot\P^0[Q',D_2](k-\ell)
\\ & +
\P^0[Q',D_1](\ell)\cdot\P^1[Q',D_2](k-\ell) 
\end{align*}
(i.e., one of the databases should have precisely one answer and the other no answers at all).

\subsubsection{Computing $\sumword_k(A,D)$ when $Q$ is connected.}
Next, we will show how to compute $\sumword_k(A,D)$, where $A=\aggduplicates\circ\tau\circ Q$, in the case where $Q$ is \sqhierarchical and connected. In this case, every free variable must be a root variable. Hence, every atom contains all free variables. In this case, we can determine, for every relation symbol $R$ of $Q$ and $R$-fact $f$, which answer $\vec t\in Q(D)$ the fact $f$ contributes to (if any), and which value $a=\tau(\vec t)$
is obtained using the fact $f$, simply by computing $\tau(\vec t)$. If $f$ yields the $\tau$ value $a$, then we call $f$ an \e{$a$-fact}. With this notation, we can compute $\sumword_k(A,D)$ using the dynamic-programming algorithm of \Cref{fig:dup-alg}.

The algorithm computes $\sumword_k(A',D)$ where $A'$ is the same as $A$, except that the aggregate function is $\aggnoduplicates\defeq 1-\aggduplicates$; hence, we compute, for each $k$, the number of subsets $E$ such that $\tau\circ Q(D\exo\cup E)$ has no duplicates (including the case where
$Q(D\exo\cup E)$ is empty). Then, we have 
$$\sumword_k(A,D)=\binom{|D\endo|}{k}-\sumword_k(A',D)\,.$$

\begin{figure}
\hrule
\begin{minipage}{\linewidth}
\renewcommand{\arraystretch}{1.1}
\begin{tabular}{ll}
\multicolumn{2}{l}{
\em Computing $\sumword_k(A',D)$ where $A'=\aggnoduplicates\circ\tau\circ Q$ and $Q$ is a connected \sqhierarchical CQ}\\
\textbf{Input:} & Database $D$ \\
\textbf{Goal:} & $\sumword_k(A,D)$ for all $k=0,\dots,|D\endo|$ 
\end{tabular}
\end{minipage}
\hrule
\smallskip
\begin{algorithmic}
\State $\mathrm{Answer}_0(0)\gets
1$ \mycomment{$\mathrm{Answer}_0(k)$ refers to the empty database}
\For{$k=1$ to $|D\endo|$}
\State $\mathrm{Answer}_0(k)\gets
0$
\EndFor
\State $\set{a_1,\dots,a_m}\gets \set{\tau(\vec t)\mid \vec t\in Q(D)}$
\For{$i=1$ to $m$}
\State $D_i\gets\set{f\in D\mid \mbox{$f$ is an $a_i$-fact}}$
\For{$k=1$ to $|D\endo|$}
\State $\mathrm{Answer}_i(k)\gets
\sum_{\ell=0}^k 
\mathrm{Answer}_{i-1}(\ell)\cdot
\left(\P^0[Q,D_i](k-\ell)+\P^1[Q,D_i](k-\ell)\right)
$
\EndFor
\EndFor
\State \textbf{return} 
$\mathrm{Answer}_m$
\end{algorithmic}
\hrule
\caption{Algorithm for computing $\sumword_k(A',D)$ for a connected \sqhierarchical CQ, a localized $\tau$, and the aggregate function no-duplicates.\label{fig:dup-alg}}
\Description{Algorithm for computing $\sumword_k(A',D)$ for a connected \sqhierarchical CQ, a localized value function, and the aggregate function no-duplicates.}
\end{figure}

The algorithm uses dynamic programming, as follows. Let $\set{a_1,\dots,a_m}$ be the set of all possible values of $\tau$ over subsets of $D$. Note that this set is simply $\set{\tau(\vec t)\mid \vec t\in Q(D)}$, since $Q$ is a CQ (hence, monotonic). 
For $i=1,\dots,m$, let $D_i$ be the subset of $D$ that consists of all $a_i$-facts. The algorithm computes 
$\sumword_k(A',D_1\cup\dots\cup D_i)$ for $i=0,\dots,m$ in dynamic programming, by increasing $i$. For $i=0$ (i.e., the empty database), the answer is $1$ if $k=0$, and $0$ otherwise. For $i>0$, a set of endogenous facts from $D_1\cup\dots\cup D_i$ contributes $1$ if and only if it is the union of a set of endogenous facts from $D_1\cup\dots\cup D_{i-1}$ with no duplicates (hence, we use the previously computed value) and a subset of $D_i$ where there is at most one answer. For the latter value, we use $\P^0$ and $\P^1$ that we discussed in \Cref{sec:P0-and-P1}.
Hence, the algorithm computes 
$\sumword_k(A',D_1\cup\dots\cup D_i)$ as
$$
\sum_{\ell=0}^k 
\sumword_\ell(A',D_1\cup\dots\cup D_{i-1})\cdot
\left(\P^0[Q,D_i](k-\ell)+\P^1[Q,D_i](k-\ell)\right)
$$
This concludes the computation in the case where $Q$ is connected.

\subsubsection{Computing $\sumword_k(A,D)$ when $Q$ is disconnected.} Finally, we consider the case where $Q$ is the Cartesian product of two \sqhierarchical CQs $Q_1$ and $Q_2$. Without loss of generality, assume that $\tau$ is localized on an atom of $Q_1$, and that $Q_1$ is connected (otherwise, we transfer its disconnected parts to $Q_2$). 

Let $D_1$ and $D_2$ be the subsets of $D$ that consist of facts of the relations of $Q_1$ and $Q_2$, respectively. Let $A_1$ be the query $A$ where $Q$ is replaced with $Q_1$.

A subset $E$ of $D\endo$ 
has the form $E=E_1\cupdot E_2$ where $E_1\subseteq D_1$ and $E_2\subseteq D_2$. Then $E$ is such that $\tau\circ Q(D\cup E)$ has duplicates if and only if one of the following mutually-exclusive options holds:
\begin{enumerate}
    \item $Q_1(D_1\exo\cup E_1)$ is nonempty and $Q_2(D_2\exo\cup E_2)$ has at least two answers (hence, every answer of $\tau\circ Q$ is duplicated);
    \item $Q_1(D\cup E_1)$ has duplicates and $Q_2(D\cup E_2)$ has precisely one answer.
\end{enumerate}
Consequently, we get the following:
\begin{align*}
\sumword_k(A,D)= &
\sum_{\ell=0}^k
\left(
\binom{|D_1\endo|}{\ell}
-\P^0[Q_1,D_1](\ell)\right)
\cdot
\left(\binom{|D_2\endo|}{k-\ell}
-\P^0[Q_2,D_2](k-\ell)
-\P^1[Q_2,D_2](k-\ell)
\right) \\ &
+ 
\sumword_\ell(A_1,D_1)\cdot \P^1[Q_2,D_2](k-\ell)
\end{align*}

This concludes the description of the algorithm for computing $\sumword_k(A,D)$, hence proving the tractability of $\Shapley(f, A)$.

\section{Omitted Details from \Cref{sec:finer}}
We now give the missing proofs and details from \Cref{sec:finer}. We begin by collecting some helpful insights that will help us prove \Cref{thm:hardness-monotonic}. Throughout this section, we will view functions $\gamma \colon \reals \to \reals$ interchangeably as operating over a number (resulting in a number) or operating over a bag of numbers (pointwise, resulting in a bag of numbers).

\begin{observation}
Let $B=(X,\mu)$ be a bag of real numbers, $\alpha$ an aggregation function and $\gamma \colon \reals \to \reals$. 
    \begin{itemize}
        \item If $\alpha \in \set{\aggmin, \aggmax}$ and $\gamma$ is monotonically increasing, then 
        $(\alpha \circ \gamma)(B) = (\gamma \circ \alpha)(B)$
        \item If $\alpha = \aggquantile_q$ for some $q \in (0,1)$ and $\gamma$ is monotonically increasing define $\gamma'$ by $\gamma'(x) = \gamma(x)$ for all $x \in X$ with $\mu(x) > 0$ and by linear interpolation between these points (and $0$ elsewhere). Then
        $(\alpha \circ \gamma)(B) = (\gamma' \circ \alpha)(B)$.
    \end{itemize}
\end{observation}

\begin{lemma}\label{lem:additive-aggr}
    Let $B=(X,\mu)$ be a bag of real numbers, $\alpha$ an aggregation function and $\gamma_1,\gamma_2 \colon \reals \to \reals$ be monotonically increasing. Then, for $\alpha \in \set{\aggmax, \aggmin, \aggavg, \aggquantile_q}$, we have $(\alpha \circ (\gamma_1 + \gamma_2)) (B) = (\alpha \circ \gamma_1) (B) + (\alpha \circ \gamma_2) (B) 
    $.
\end{lemma}

\begin{proof}
    For $\alpha = \aggavg$, this follows directly from the definition. For the other aggregation functions $\alpha$, it follows from the previous observation.
\end{proof}

\begin{observation}\label{obs:injective}
    Let $Q$ be a CQ without self-joins, $1 \leq i \leq \arity(Q)$ and let $x_i$ is denote the $i$th variable in the head of $Q$. Furthermore, let $\gamma \colon \reals \to \reals$ be injective. 
    Consider the following transformation $\pi$ of facts in the schema of $Q$:
    For each fact $f = R(\vec a)$ let $\pi(f)$ be the fact $\pi(f) = R(\vec a')$ defined by
    \[
    \vec a'[j] \defeq \begin{cases}
        \gamma(\vec a[j]), &\text{ if } x_i \text{ appears at position } j \text{ in the atom } R(\vec x) \text{ of } Q,\\
        \vec a[j], &\text{ otherwise,}
    \end{cases}
    \]
    Then, for each database $D$, we have $(\gamma \circ \tauidi{i} \circ Q)(D) = (\tauidi{i} \circ Q)(\pi(D))$.
\end{observation}

Now, we are ready to prove \Cref{thm:hardness-monotonic}.

\thmhardnessmonotonic*

\begin{proof}
    Let $A_1 = \alpha \circ \gammamon \circ \tauidi{i} \circ Q$ and $A_2 = \alpha \circ \tauidi{i} \circ Q$. We aim to show that we can compute $\Shapley(f,A_1)(D)$ by computing $\Shapley(f,A_2)(D)$ and $\Shapley(f',A_2)(D')$ for some $f'$ and $D'$ constructible in polynomial time. To that end, 
    consider a third query $A_3 = \alpha \circ (\gammamon + \tauid) \circ \tauidi{i} \circ Q$. For every $D$, we can apply \Cref{lem:additive-aggr} to the bag $B = (\tauidi{i} \circ Q)(D)$ to obtain
    \begin{align*}
    A_3(D) &=  (\alpha \circ (\gammamon + \tauid)) ((\tauidi{i} \circ Q)(D)) \\
    &= (\alpha \circ \gammamon) ((\tauidi{i} \circ Q)(D)) + (\alpha \circ \tauid) ((\tauidi{i} \circ Q)(D)) \\
    &= A_1(D) + A_2(D)    
    \end{align*}
    By the linearity of the Shapley value, this implies
    \[
    \Shapley(f,A_1)(D) = \Shapley(f,A_3)(D) - \Shapley(f,A_2)(D).
    \]
    Finally, we apply \Cref{obs:injective} using the function $\gamma = \gammamon + \tauid$. This yields 
    \[
    \Shapley(f, A_1)(D) =\Shapley(\pi(f),A_2)(\pi(D)) - \Shapley(f,A_2)(D),
    \]
    and since $\pi$ can be computed in polynomial time, this shows the claim.
\end{proof}

Finally, we show \Cref{prop:localizations-in-ptime}.

\proplocalizationsinptime*

\begin{proof}
For the first query $A_1 = \aggavg \circ \taurelui{2} \circ \Qxyyz$, let us have a closer look at $\Qxyyz(x,z)$. This query is disconnected, so its set of answers is the Cartesian product of the answers of its two connected components $Q_1(z) \la T(z)$ and $Q_2(x) \la R(x,y)S(y)$.  The bag $(\taurelui{2} \circ \Qxyyz)(D)$ is therefore obtained from the bag $(\taurelu \circ Q_1)(D)$ by multiplying each multiplicity by $(\aggcount \circ Q_2)(D)$. Since multiplying each multiplicity of a bag by a number $>0$ does not change its average, we can conclude
\[
A_1(D) = (\aggavg \circ \taurelu \circ Q_1)(D) \cdot Q_2(D),
\]
where we interpret the Boolean query $Q_2$ to be numerical by interpreting $\mathsf{True}$ as $1$ and $\mathsf{False}$ as $0$. 
Since only the facts of the form $T(\vec a)$ are relevant for $Q_1$ and only the other facts are relevant for $Q_2$, we can split $D$ into databases $D_{\set{T}}$ and $D_{\set{R,S}}$ of facts relevant for $Q_1$ and $Q_2$, respectively. That way, we obtain
\[\sumword_k(A_1, D) = \sum_{\ell=0}^k \sumword_\ell(\aggavg \circ \taurelu \circ Q_1,D_{\set{T}}) \cdot \sumword_{k - \ell}(Q_2,D_{\set{R,S}}).\]
Since $Q_1$ is \qhierarchical and $Q_2$ is \exhierarchical, we can compute $\sumword_\ell(\aggavg \circ \taurelu \circ Q_1,D_{\set{T}})$ and $\sumword_{k - \ell}(Q_2,D_{\set{R,S}})$ in polynomial time, which completes the proof for $A_1$.

For the second query $A_2 = \aggmed \circ \taugbi{0}{2} \circ \Qxyyz$, we can argue analogously, since the median of a bag is preserved when we multiply each multiplicity by a number $>0$.\footnote{Note that this statement is not true for other quantiles.}

Finally, let $D$ be an input database to $A_3 = \aggduplicates \circ \tauidi{2} \circ \QFxyy$. If we simply wanted to answer $A_3$, we could observe that $y$ appears in every atom of $\QFxyy$, so we could decompose $D$ into subdatabases for the possible values of $y$, answer the query on each of them separately and then take the maximum of the obtained answers.
We now show how to use this approach to compute $\sumword_k(A_3,D)$ for each $k$ in polynomial time. Let $b_1, \ldots, b_t$ be the different values in $\tauidi{2} \circ \QFxyy$ with positive multiplicity and let $D_i$ be the subdatabase of $D$ consisting of $S(b_i)$ and all the facts of the form $R(a,b_i)$. Furthermore, let $D_{t+1}$ consist of all the remaining facts, that is, facts that are not part of the image of any homomorphism from $\QFxyy$ to $D$.

Now, for $1 \leq i \leq t$, a $k$-subset $E \subseteq D_i\endo$ has $A_3(E \cup D_i\exo)$ if and only if $E \cup D_i\exo$ contains $S(b_i)$ and at least two facts of the form $R(a,b_i)$. To find the number of such $k$-subsets, let $\delta_i = 1$ if $S(b_i)$ is endogenous and $0$ otherwise, and let $\varepsilon_i$ be the number of exogenous facts of the form $R(a, b_i)$. Then, we need to pick $S(b_i)$ which determines $\delta_i$ of the $k$ facts and since the remaining $k - \delta_i$ facts are necessarily of the form $R(a,b_i)$, we can pick any of them as long as $k - \delta_i + \varepsilon_i \geq 2$. This shows that
\[
\sumword_k(A_3, D_i) = \begin{cases}
    \binom{\size{D_i\endo} - \delta_i}{k - \delta_i}, &\text{ if } k + \varepsilon_i - \delta_i \geq 2, \\
    0, &\text{ otherwise,}
\end{cases}
\]
for $1 \leq i \leq t$ and, of course, $\sumword_k(A_3, D_{t+1}) = 0$.
To compute $\sumword_k(A_3, D)$, it is more convenient to argue about $\sumword_k(1 - A_3, D) = \binom{\size{D\endo}}{k} - \sumword_k(A_3, D)$. To combine the results for $D_i$ step by step, let $D_i ' = D_1 \cup \ldots \cup D_i$ for $1 \leq i \leq t+1$. Since each of the $D_i$ belongs to a different value in the output (or to none at all in the case of $D_{t+1}$), we have $(1- A_3)(D_{i+1}') = (1- A_3)(D_{i}') \cdot (1- A_3)(D_{i+1})$, which implies
\[
\sumword_k(1 - A_3, D_{i+1}') = \sum_{\ell = 1}^k \sumword_\ell(1 - A_3, D_i') \cdot \sumword_{k-\ell}(1 - A_3, D_{i+1}) \,.
\]
Using this equation to iteratively compute $\big(\sumword_k(1 - A_3, D_{i}')\big)_{0 \leq k \leq \size{D\endo}}$ for $i = 1, \ldots, t+1$ yields a polynomial time algorithm to compute $\sumword_k(A_3,D)$ and hence $\Shapley(f,A_3)(D)$. 
\end{proof}

\end{document}